\documentclass[11pt]{article}
\usepackage{thmtools}
\usepackage{thm-restate}
\usepackage{bm}
\usepackage{mathrsfs}           %
\usepackage{tikz}
\usetikzlibrary{positioning,chains,fit,shapes,calc}
\usepackage{algorithm}
\usepackage{algpseudocode}
\usepackage{mathtools}

\usepackage{subcaption}
\usepackage{enumerate}

\usepackage{amsmath,amssymb, amsthm}    %
\usepackage{verbatim}           %
\usepackage{xspace}             %
\usepackage{graphicx,float}     %
\usepackage{ifthen,calc}        %
\usepackage{textcomp}           %
\usepackage{fancybox}           %
\usepackage{hhline}             %
\usepackage{float}              %

\usepackage[vflt]{floatflt}     %
\usepackage[small,compact]{titlesec}
\usepackage{setspace}

\usepackage{enumitem}
\setenumerate[1]{label={(\arabic*)}}

\usepackage{color}
\definecolor{ForestGreen}{rgb}{0.1333,0.5451,0.1333}
\usepackage{thumbpdf}
\usepackage[letterpaper,
colorlinks,linkcolor=ForestGreen,citecolor=ForestGreen,
bookmarks,bookmarksopen,bookmarksnumbered]
{hyperref}
\usepackage{cleveref}
\usepackage[margin = 1in]{geometry}
\newcommand{\showccc}[0]{0}
\newcommand{\ccc}[2][nothing]{%
	\ifthenelse{\showccc=0}{}{
		\ensuremath{^{\Lsh\Rsh}}\marginpar{\raggedright\tiny\textsf{%
				\ifthenelse{\equal{#1}{nothing}}{}{\textbf{#1}\\}#2}}}}
\newtheorem{theorem}{Theorem}
\newtheorem{claim}{Claim}
\newtheorem{proposition}{Proposition}
\newtheorem{corollary}{Corollary}
\newtheorem{definition}{Definition}

\newtheorem{lemma}{Lemma}
\newtheorem{fact}{Fact}

\newtheorem{assumption}{Assumption}
\newtheorem{problem}{Problem}

\newcommand{\defeq}{:=}
\newcommand{\norm}[1]{\left\lVert#1\right\rVert}

\newcommand{\inprod}[2]{\left\langle#1, #2\right\rangle}

\newcommand{\eps}{\epsilon}
\newcommand{\lam}{\lambda}

\newcommand{\R}{\mathbb{R}}

\newcommand{\N}{\mathbb{N}}

\newcommand{\half}{\frac{1}{2}}

\newcommand{\1}{\mathbf{1}}
\newcommand{\E}{\mathbb{E}}
\newcommand{\Var}{\textup{Var}}

\newcommand{\Nor}{\mathcal{N}}

\newcommand{\Tr}{\textup{Tr}}

\newcommand{\xset}{\mathcal{X}}

\newcommand{\ma}{\mathbf{A}}

\newcommand{\id}{\mathbf{I}}

\definecolor{burntorange}{rgb}{0.8, 0.33, 0.0}

\newcommand{\tO}{\widetilde{O}}

\newcommand{\Par}[1]{\left(#1\right)}
\newcommand{\Brack}[1]{\left[#1\right]}
\newcommand{\Brace}[1]{\left\{#1\right\}}
\newcommand{\Abs}[1]{\left|#1\right|}
\newcommand{\Cov}{\textup{Cov}}

\newcommand{\alg}{\mathcal{A}}

\newcommand{\AlternateSample}{\texttt{AlternateSample}}
\newcommand{\K}{\mathcal{K}}
\newcommand{\hpi}{\hat{\pi}}

\newcommand{\T}{\mathcal{T}}
\newcommand{\tvd}[2]{\norm{#1 - #2}_{\textup{TV}}}
\newcommand{\Npsi}{\mathcal{D}^{\vhi}}
\definecolor{ruoqi}{rgb}{0.2, 0.3, 0.7}

\newcommand{\daogao}[1]{\textcolor{blue}{\textbf{daogao:} #1}}

\newcommand{\Inner}{\texttt{InnerLoop}}

\newcommand{\med}{\textup{med}}
\newcommand{\tpi}{\tilde{\pi}}

\newcommand{\vhi}{\varphi}

\newcommand{\brho}{\bar{\rho}}

\newcommand{\br}{\bar{r}}
\newcommand{\jia}{j_{i,b}}
\newcommand{\xsetd}{\xset^*}
\newcommand{\normx}[1]{\norm{#1}_{\xset}}
\newcommand{\normxd}[1]{\norm{#1}_{\xsetd}}
\newcommand{\dkl}{D_{\textup{KL}}}
\newcommand{\prob}{\mathcal{P}}
\newcommand{\data}{\mathcal{D}}
\newcommand{\mech}{\mathcal{M}}

\newcommand{\msig}{\boldsymbol{\Sigma}}
\newcommand{\sset}{\mathcal{S}}
\newcommand{\dd}{\textup{d}}
\newcommand{\Ferm}{F_{\textup{erm}}}
\newcommand{\Fpop}{F_{\textup{sco}}}

\newcommand{\mx}{\mathbf{X}}
\newcommand{\my}{\mathbf{Y}}
\newcommand{\mm}{\mathbf{M}}
\newcommand{\ind}{\mathcal{I}}
\newcommand{\jset}{\mathcal{J}}
\newcommand{\mmu}{\mathbf{U}}
\newcommand{\mb}{\mathbf{B}}
\newcommand{\bbA}{\mathbb{A}}
\newcommand{\bbB}{\mathbb{B}}
\newcommand{\bP}{\overline{P}}

\DeclarePairedDelimiterX{\xdivergence}[2]{(}{)}{%
	#1\;\delimsize\|\;#2%
}

\usepackage{float}
\floatstyle{plain}\newfloat{myfig}{t}{figs}[section]
\floatname{myfig}{\textsc{Figure}}
\floatstyle{plain}\newfloat{myalg}{H}{algs}[section]
\floatname{myalg}{}
\newcommand{\cN}{\mathcal N}

\newcommand{\FP}{F_{\prob}}
\newcommand{\cV}{\mathcal{V}}
\newcommand{\sign}{\mathrm{sign}}
\newcommand{\hv}{\hat{v}} 
\usepackage{url}

\begin{document}

	\begin{titlepage}
		\def\thepage{}
		\thispagestyle{empty}
		
		\title{Algorithmic Aspects of the Log-Laplace Transform \\ and a Non-Euclidean Proximal Sampler} 
		
		\date{}
		\author{
			Sivakanth Gopi\thanks{Microsoft Research, {\tt sigopi@microsoft.com}}
			\and
			Yin Tat Lee\thanks{Microsoft Research, {\tt yintatlee@microsoft.com}}
			\and
			Daogao Liu\thanks{University of Washington, {\tt dgliu@uw.edu}}
			\and
			Ruoqi Shen\thanks{University of Washington, {\tt shenr3@cs.washington.edu}}
			\and
			Kevin Tian\thanks{Microsoft Research, {\tt tiankevin@microsoft.com}}
		}
		
		\maketitle

\abstract{
The development of efficient sampling algorithms catering to non-Euclidean geometries has been a challenging endeavor, as discretization techniques which succeed in the Euclidean setting do not readily carry over to more general settings. We develop a non-Euclidean analog of the recent proximal sampler of \cite{lee2021structured}, which naturally induces regularization by an object known as the log-Laplace transform (LLT) of a density. We prove new mathematical properties (with an algorithmic flavor) of the LLT, such as strong convexity-smoothness duality and an isoperimetric inequality, which yield a mixing time on our proximal sampler that matches known rates in the Euclidean setting up to a quadratic factor, under a warm start. As our main application, we demonstrate the implications of our framework (and potential improvements thereof) towards achieving near-optimal zeroth-order query complexity-excess risk tradeoffs for differentially private convex optimization in $\ell_p$ and Schatten-$p$ norms for $p \in [1, 2]$. We find our investigation of the LLT to be a promising proof-of-concept of its utility as a tool for designing samplers, and outline directions for future exploration.
}
 		
	\end{titlepage}

\section{Introduction}
\label{sec:intro}

The development of samplers for continuous distributions, under weak oracle access to the corresponding densities, has seen a flurry of recent research activity. For applications in settings inspired by machine learning or computational statistics, this development has in large part built upon connections between sampling and continuous optimization. Inspired by perspectives on sampling as optimization in the space of measures \cite{JordanKO98} and starting with pioneering work of \cite{dalalyan2017theoretical}, a long sequence of results, e.g.\ \cite{Dalalyan17, ChengCBJ18, DwivediCW019, durmus2019high, ChenV19, DurmusMM19, ShenL19, ChenDW020, LeeST20, ChewiLACGR21}, has used analysis techniques from convex optimization to bound the convergence rates of sampling algorithms for densities. We refer the reader to the survey \cite{Chewi23} for a more complete account, but note in almost all cases, the focus has been on sampling from densities satisfying regularity assumptions stated in the Euclidean ($\ell_2$) norm, e.g.\ $\ell_2$-bounded derivatives.

The theory of continuous optimization under regularity assumptions stated for non-Euclidean geometries has played an important role in algorithm design. These geometries naturally arise when the optimization problem is over a structured constraint set, such as an $\ell_p$ ball or a polytope. In diverse applications such as learning from experts \cite{AroraHK12}, sparse recovery \cite{CandesRT06}, multi-armed bandits \cite{BubeckC12}, matrix completion \cite{AgarwalNW10}, fair resource allocation \cite{DiakonikolasFO20}, and robust PCA \cite{Jambulapati0T20}, first-order mirror descent techniques for $\ell_p$ or Schatten-$p$ geometries have been a remarkable success story. Beyond these applications, the theory of self-concordant barriers (and the Riemannian geometries induced by their Hessians) has been greatly influential to the theory of convex programming and interior point methods \cite{nesterov2002riemannian, nemirovski2004interior}.\footnote{Self-concordance requires that the second derivative of a function is stable to perturbations which are measured in the induced norm. For notation and definitions used throughout the paper, see Section~\ref{sec:prelims}.}

\paragraph{Non-Euclidean samplers.} A natural direction for building the theory of logconcave sampling (the analog of convex optimization) is thus to develop samplers which can handle non-Euclidean regularity assumptions and constraint sets. Unfortunately, progress in this direction has relatively lagged behind optimization counterparts, as discretization tools which work well in the Euclidean case do not readily generalize. Briefly (with an extended discussion deferred to Section~\ref{ssec:prior}), most prior attempts at giving non-Euclidean samplers have focused on analyzing variants of the \emph{mirrored Langevin dynamics}, building upon the ubiquitous mirror descent algorithm in optimization \cite{NemirovskiY83}. The key idea of mirror descent is to choose a regularizer $\phi: \xset \to \R$ over a constraint set $\xset$, such that $\phi$ is strongly convex in an appropriate (possibly non-Euclidean) norm $\normx{\cdot}$. The regularizer $\phi$ is then used to define iterative methods for optimizing functions $f$ with regularity in $\normx{\cdot}$.

The sampling analog of this non-Euclidean generalization is to extend the \emph{Langevin dynamics}, a stochastic process inherently catered to the $\ell_2$ geometry, to use Brownian motion reweighted by the Hessian of a regularizer $\phi$. This process, which we call the mirrored Langevin dynamics (MLD), was introduced recently by \cite{ZhangPFP20} (see also \cite{HsiehKRC18} for an earlier incarnation). Several follow-up works attempted to bound convergence rates for discretizations of the MLD process, e.g.\ \cite{AhnC21, Jiang21, LiTVW22}. Unfortunately, many of these analyses have imposed rather strong conditions on $\phi$ beyond strong convexity, e.g.\ a ``modified self-concordance'' assumption used in \cite{ZhangPFP20, Jiang21, LiTVW22} which (to our knowledge) is not known to be satisfied by standard regularizers. Even more problematically, these analyses (as well as an empirical evaluation by \cite{Jiang21}) suggest that without strong relative regularity assumptions between the target density and $\phi$, na\"ive discretizations of MLD inherently do not converge to the target even in the limit. A notable exception is the work of \cite{AhnC21}, which circumvented both issues (the modified self-concordance assumption and a biased limit) using a different MLD discretization; however, it is not always clear that this discretization is feasible for standard choices of $\phi$ and $\xset$.

An alternative to directly discretizing MLD is to use a filter to control bias, akin to the MALA or Metropolized HMC algorithms which are well-studied in the Euclidean case \cite{Besag94, RobertsT96, BouRabeeH12, DwivediCW019, ChenDW020, LeeST20}. However, here too generalizing existing analyses runs into obstacles: for example, typical analyses of MALA and Metropolized HMC rely on bounding the conductance of random walks via isoperimetric inequalities on the target distribution. Prior isoperimetry bounds appear to be tailored to the $\ell_2$ geometry and properties of Gaussians (the basic strongly logconcave distribution in Euclidean settings). Potentially due to this difficulty, to our knowledge no general-purpose extension of MALA or its variants to non-Euclidean norms exists in the literature.\footnote{We mention that in certain geometries induced by structured manifolds (discussed in part in Section~\ref{ssec:prior}), generalizations of MALA or Metropolized HMC have been previously proposed, e.g.\ \cite{GirolamiC11, Barp20}. These works are motivated by related, but different, settings to the ones considered in this work (we mainly study norm regularity, akin to first-order convex optimization), and their focus is not on establishing non-asymptotic mixing time bounds.}

\paragraph{Proximal samplers.} In this paper, we overcome these difficulties by following a third strategy for the design of efficient samplers: a proximal approach recently proposed by \cite{lee2021structured}. To sample from a density $\pi$ on $\R^d$ proportional to  $\exp(-f)$, the algorithm of \cite{lee2021structured} first extends the space to $\R^d \times \R^d$, and defines a joint density $\hpi$ such that, for some parameter $\eta > 0$,
\begin{equation}\label{eq:l2_joint}\dd\hpi(z) \propto \exp\Par{-f(x) - \frac 1 {2\eta}\norm{x - y}_2^2} \dd z\text{ where } z = (x, y) \in \R^d \times \R^d.\end{equation}
It is straightforward to see that for any $\eta$, the $x$-marginal of $\hpi$ is the original distribution $\pi$, and further \cite{lee2021structured} shows that alternating sampling from the conditional distributions of $\hpi$, i.e.\ $\hpi(x \mid y)$ or $\hpi(y \mid x)$, mixes rapidly. We give an extended discussion on recent activity on designing and harnessing proximal samplers building upon \cite{lee2021structured} in Section~\ref{ssec:prior}, but mention that instantiations of the framework have resulted in state-of-the-art runtimes for many structured density families \cite{ChenCSW22, liang2022, GLL22}. Motivated by the success of proximal methods in the Euclidean setting, one goal of our work is to extend this technique to non-Euclidean geometries.

\paragraph{Our approach.} Our main insight is that a generalization of the strategy in \cite{lee2021structured} induces a well-studied object in probability theory called the \emph{log-Laplace transform} (LLT). Letting $\vhi: \R^d \to \R$ be a convex function in the dual space $y \in \R^d$, our generalization of \eqref{eq:l2_joint} defines the joint density
\begin{equation}\label{eq:gen_joint}
	\begin{aligned}
\dd\hpi(z) &\propto \exp\Par{-f(x) + \Par{\inprod{x}{y} - \vhi(y) - \psi(x)}} \dd z,\\
\text{where } \psi(x) &\defeq \log\Par{\int \exp\Par{\inprod{x}{y} - \vhi(y)} \dd y}. 
\end{aligned}
\end{equation}
The function $\psi$ is called the LLT of $\vhi$, and it has an interpretation as a normalizing constant for induced densities $\Npsi_x$ on the dual space proportional to $\exp(\inprod{x}{\cdot} - \vhi)$. Indeed, $\Npsi_x$ is defined exactly so the $x$-marginal of $\hpi$ is $\pi \propto \exp(-f)$. When $\eta = 1$ and $\vhi, \psi$ are quadratics, this is exactly \eqref{eq:l2_joint}; we discuss the case of general $\eta$ in Section~\ref{ssec:techniques}. Moreover, the LLT is a well-studied mathematical object: it arises in probability theory as a \emph{cumulant-generating function}, i.e.\ derivatives of the LLT yield cumulants of the induced distributions $\Npsi_x$, just as derivatives of the MGF yield moments. 

The LLT famously appeared in Cram\'er's theorem on large deviations \cite{Cramer38}, and its cumulant-generating properties have yielded fundamental concentration results in convex geometry \cite{Klartag06, EldanK11, KlartagM12}. More recently, algorithmically-motivated properties of the LLT have been studied in settings such as optimization \cite{bubeck2014entropic}, where it was used to define an optimal self-concordant barrier, as well as connections to localization schemes for sampling from discrete distributions \cite{ChenE22}. 

We continue this investigation by demonstrating new mathematical properties of the LLT with an algorithmic flavor, and showcasing uses of the LLT as a tool for continuous logconcave sampling. In particular, armed with a deeper understanding of the LLT, we overcome several of the aforementioned barriers to non-Euclidean sampler design and develop a generalized proximal sampler. We further give applications of our sampler to obtain new complexity results for non-Euclidean differentially private convex optimization, building upon a connection discovered by \cite{GLL22, GopiLLST23}. We are optimistic that the LLT will find additional uses in sampler design (potentially beyond the proximal sampling framework, building upon the new properties we prove), and suggest a number of avenues of future exploration to the community in Section~\ref{sec:future}.

\subsection{Our results}\label{ssec:results}

In this section, we overview our results, which separate cleanly into three categories.

\paragraph{Algorithmic aspects of the LLT.} It is well-known that the derivatives of the LLT at a point $x \in \R^d$ are \emph{cumulants} of the induced density on $y \in \R^d$:
\[\dd \Npsi_x(y) \propto \exp\Par{\inprod{x}{y} - \vhi(y)} \dd y.\]
For example, $\nabla \psi(x) = \E_{y \sim \Npsi_x}[y]$, and $\nabla^2\psi(x)$ is the covariance of $\Npsi_x$. Further, it was shown in \cite{bubeck2014entropic} that if $\psi$ is the LLT of a convex function $\vhi$, then $\psi$ is convex and self-concordant. Building upon these facts, in Section~\ref{sec:llt}, we prove the following new properties of the LLT.

\begin{itemize}
	\item \textit{Strong convexity-smoothness duality.} Let $\norm{\cdot}$ be a norm on $\R^d$. We prove that if $\vhi: \R^d \to \R$ is $L$-smooth in the dual norm $\norm{\cdot}_*$, its LLT $\psi: \R^d \to \R$ is $\frac 1 {L}$-strongly convex in $\norm{\cdot}$.\footnote{The constant factor $1$ here is optimal, as demonstrated by quadratics.} This fact parallels a similar, well-known form of strong convexity-smoothness duality for Fenchel conjugates \cite{Shalev-Shwartz07, KakadeST09}. Our proof does not require $\vhi$ to be convex. We further show that the converse holds as well: a $\frac 1 L$-strongly convex $\vhi$ has a $L$-smooth LLT.
	\item \textit{Isoperimetry in the Hessian norm.} We prove a one-dimensional isoperimetric inequality for densities of the form $\exp(-\phi)$, where $\phi: \R \to \R$ is self-concordant and convex. By appealing to (a strong variant of) the localization lemma of \cite{lovasz1993random}, this proves that measures which are strongly logconcave with respect to convex and self-concordant $\phi: \R^d \to \R$ satisfy a similar isoperimetric inequality in the Riemannian geometry induced by $\nabla^2 \phi$. Importantly, due to self-concordance of the LLT, this applies to strongly logconcave measures in an LLT.
	\item \textit{Overlap of induced distributions $\Npsi_x$.} We provide a KL divergence bound on the distributions $\Npsi_x$ and $\Npsi_{x'}$ for $x$ and $x'$ which are close in the Riemannian distance induced by $\psi$. Combined with our isoperimetric inequality and a classical argument of \cite{DyerFK91}, this proves a lower bound on the conductance of an alternating sampler for densities of the form \eqref{eq:gen_joint}.
\end{itemize}

These new properties of the LLT suggest that it may find uses in designing samplers under non-Euclidean geometries beyond those explored in Sections~\ref{sec:sampling} and~\ref{sec:apps} of our paper. For example, the LLT of a smooth function is strongly convex and self-concordant, which are exactly the properties required by the mirror Langevin discretization scheme of \cite{AhnC21}. In optimization, regularizers $\phi$ for mirror descent typically only require strong convexity (and not self-concordance). However, controlling the evolution of the geometry induced by $\nabla^2 \phi$ is critical for discretizing MLD schemes, so imposing self-concordance (as opposed to more non-standard regularity such as the modified self-concordance of \cite{ZhangPFP20, Jiang21, LiTVW22}) may be viewed as a minimal assumption. Problematically, standard strongly convex regularizers for mirror descent such as entropy or $\ell_p^2$ \emph{are not} self-concordant, so LLTs are a way of bridging this gap for sampling. Moreover, our new isoperimetric inequality and conductance bounds suggest that LLTs may find use in Metropolized sampling schemes, paving the way for non-Euclidean generalizations of MALA and its variants. 

In some sense, our new duality result is a generic way of taking a strongly convex regularizer and transform it, via the \emph{Fenchel transform} and the \emph{log-Laplace transform}, to another regularizer which is strongly convex in the same norm, but also self-concordant. The first transform makes the function smooth in the dual \cite{KakadeST09}, and the second effectively undoes this change. We will later discuss an application of this framework in improving the oracle complexity of the problem of private stochastic convex optimization in the $\ell_p$ geometry, using the LLT of the $\ell_q^2$ regularizer.

\paragraph{Non-Euclidean proximal sampling.} In Section~\ref{sec:sampling}, we build upon these aforementioned tools to analyze the mixing time of an alternating scheme for sampling densities $\pi$ on convex, compact $\xset \subset \R^d$ equipped with a norm $\normx{\cdot}$, where $\pi$ is proportional to $\exp\Par{-F(x) - \eta \mu \psi(x)}\1_{\xset}(x)$. Here, $F: \xset \to \R$ is convex, $\eta, \mu > 0$ are tunable parameters, and $\psi$ is the LLT of $\eta$-smooth $\vhi: \R^d \to \R$ in the dual norm $\normxd{\cdot}$. We prove in Theorem~\ref{thm:mainllt} that alternately sampling from conditional distributions of the extended density on $z = (x, y) \in \xset \times \R^d$ proportional to
\begin{equation}\label{eq:density_intro}\exp\Par{-F(x) - \eta\mu\psi(x) + \Par{\inprod{x}{y} - \vhi(y) - \psi(x)}}\1_{\xset}(x)\end{equation}
has stationary distribution $\pi$, and converges in $\approx \frac 1 {\eta^2\mu^2}$ iterations for a warm start. More specifically, the convergence rate of our sampler depends polylogarithmically on both the warmness $\beta$ of the point it is initialized with, and the inverse of the total variation error $\delta$. The form of \eqref{eq:density_intro} is the same as \eqref{eq:gen_joint}, but we impose that $f$ is $\eta\mu$-relatively strongly convex in $\psi$.

We first compare this result to the Euclidean proximal sampler of \cite{lee2021structured}, who proved a similar result for alternating sampling densities of the form \eqref{eq:l2_joint}. The main result of \cite{lee2021structured} shows that if $f$ is $\mu$-strongly convex in the $\ell_2$ norm, then alternating sampling from the marginals of \eqref{eq:l2_joint} converges in $\approx \frac 1 {\eta\mu}$ iterations, also with polylogarithmic dependence on the target total variation error. Our result can be viewed as an extension of this result; instead of requiring $\mu$-strong convexity in the $\ell_2$ norm (which is equivalent to relative strong convexity with respect to the function $x \to \half \norm{x}_2^2$), we require $\mu$-relative strong convexity in the function $\eta\psi$. In light of our duality result, $\eta\psi$ is $1$-strongly convex in $\normx{\cdot}$, so it is the natural ``unit'' for measuring strong convexity. 

The two main shortcomings of Theorem~\ref{thm:mainllt}'s rate are that it scales quadratically in $\frac 1 {\eta\mu}$, and depends polylogarithmically on the warmness. In contrast, the rate of \cite{lee2021structured} scales \emph{linearly} in $\frac 1 {\eta\mu}$, and depends \emph{doubly logarithmically} on the warmness. The latter difference is important because in many applications, explicit starting distributions have warmness bounds growing exponentially in parameters such as the dimension $d$. We refer the reader to a discussion in Section 1.1 of \cite{LeeST21} on warmness assumptions under $\ell_2$ geometry, which have created a $\approx \sqrt{d}$-sized gap on mixing time bounds for MALA, with and without a polynomially-bounded warm start \cite{ChewiLACGR21, LeeST20}. We believe it is an important future direction to close these gaps in the mixing time scaling and warmness assumptions for our sampler in Section~\ref{sec:sampling}, analogously to the result of \cite{lee2021structured}. Notably, there has been an ongoing exploration of new proof techniques for the convergence of proximal samplers by the community \cite{ChenCSW22, ChenE22}, and we are optimistic similar advancements can be made in non-Euclidean settings in future work, discussed further in Section~\ref{ssec:prior}. We view our result as a key first step towards the program of completing a theory of non-Euclidean proximal sampling.

We remark that the parameters $\eta$ and $\mu$ play different roles: $\mu$ governs the strong logconcavity of the stationary distribution, and $\eta$ controls the strong logconcavity of the $x$-conditional distribution of \eqref{eq:density_intro}, which is tuned to govern the convergence rate of sampling from the conditional distribution. In particular, we further show that when $F$ is $G$-Lipschitz in $\normx{\cdot}$, then as long as $\eta \lesssim G^{-2}$, the conditional sampling required by \eqref{eq:density_intro} can be performed in constant calls to a value oracle to $F$ in expectation. This result holds even when $F$ is a distribution over $G$-Lipschitz functions, and we only have sample access to this distribution. This extends a similar implementation of the marginal sampler required by \cite{lee2021structured} for log-Lipschitz densities in the $\ell_2$ norm, given by \cite{GLL22}. The remaining complexity of the marginal sampling depends on the structure of the chosen $\vhi$ and $\xset$, but is independent of $F$; we give a discussion of this aspect of our sampler in Sections~\ref{ssec:valueoracle} and~\ref{sec:future}.

\paragraph{Zeroth-order private convex optimization.} To highlight the potential of our techniques, in Section~\ref{sec:apps} we design LLTs based on the smoothness of the function $\vhi_q(x) = \frac {p - 1} 2 \norm{x}_q^2$ in the norm $\ell_q$, where $\frac 1 p + \frac 1 q = 1$ and $p \in [1, 2]$, $q \ge 2$. We show that the additive range of $\psi_{\eta, p}$,\footnote{We use slightly different notation than in Section~\ref{sec:apps} for convenience of exposition here.} the LLT of $\eta\vhi_q$ for $\eta \lesssim  \frac 1 d$,\footnote{This restriction is discussed further in Section~\ref{ssec:techniques}, but does not bottleneck our privacy applications.} is bounded by $O(\frac 1 {(p - 1)\eta})$ over the unit $\ell_p$ ball. This makes $\eta\psi_{\eta, p}$ competitive with the canonical choice of regularizer in $\ell_p$ norms for optimization, namely $r_p(x) \defeq \frac 1 {2(p - 1)}\norm{x}_p^2$, which has the same additive range and strong convexity parameters as $\eta\psi_{\eta, p}$ (up to constants). We further build efficient value oracles and samplers for induced densities for $\psi_{\eta, p}$ in Section~\ref{ssec:valueoracle}.

A critical difference between $\eta\psi$ and $r_p$, however, is that regularizing by a multiple of $\eta\psi$ admits efficient samplers via the machinery in Section~\ref{sec:sampling}; to our knowledge no similar technique is known for $r_p$. This difference is particularly important in the setting of \emph{differentially private convex optimization}: see Problem~\ref{prob:dpco} for a formal statement of the problem we study. Recently, \cite{GopiLLST23} showed that to privately minimize either population or empirical risk for a distribution over convex functions which are Lipschitz in a (possibly non-Euclidean) norm $\normx{\cdot}$, it suffices to sample from a regularized density $\propto \exp(-k(\Ferm + \mu r))$. Here, $\Ferm = \frac 1 n \sum_{i \in [n]} f_i$ is the empirical risk over $n$ samples $\{f_i\}_{i \in [n]}$, $k, \mu$ are tunable parameters, and $r$ is a $1$-strongly convex regularizer in $\normx{\cdot}$.

Our new sampling results show a demonstrable algorithmic advantage of using $\eta\psi_{\eta, p}$ as a regularizer for $\ell_p$ geometries, as opposed to $r_p$. In Theorem~\ref{thm:lp}, we give algorithms for private convex optimization matching the state-of-the-art excess risk bounds for private convex optimization recently attained by \cite{GopiLLST23} (who used $r_p$ as their regularizer). Under a warm start, our new algorithms further improve the \emph{value (zeroth-order) oracle} complexities of private convex optimization under $\ell_p$ regularity in dimension $d$ compared to \cite{GopiLLST23} by $\textup{poly}(d)$ factors, i.e.\ the number of queries to $\{f_i\}_{i \in [n]}$ used, in certain regimes of the sample complexity $n$. We also show these new value oracle complexities extend straightforwardly to private convex optimization over matrix spaces satisfying Schatten-$p$ norm regularity, yielding similar conditional improvements.

We note that our results match, up to a quadratic overhead, the value oracle complexities in the $\ell_2$ setting obtained by \cite{GLL22}, for all $\ell_p$ norms where $p \in [1, 2]$. In Appendix~\ref{app:lb}, we extend lower bounds for stochastic optimization from \cite{DJWW15, GLL22} to the $\ell_p$ setting to show the value oracle complexities of Theorems~\ref{thm:mainllt} and~\ref{thm:lp} are indeed within a near-quadratic factor of optimal. We find it exciting that plausible qualitative improvements to Theorem~\ref{thm:mainllt} (making it fully analogous to \cite{lee2021structured}) would close these overheads up to logarithmic factors, whereas a $\textup{poly}(d)$-factor overhead appears inherent to the previous approach of \cite{GopiLLST23}, which relied on the Euclidean proximal sampler. We give additional discussion of these points in Sections~\ref{ssec:privacy} and~\ref{sec:future}.

\subsection{Our techniques}\label{ssec:techniques}

Analogously to Section~\ref{ssec:results}, in this section we split our discussion of our techniques into three parts.

\paragraph{Algorithmic aspects of the LLT.} We first discuss our strong convexity-smoothness duality result. From a convex geometry perspective, smoothness of $\vhi$ (with LLT $\psi$) ensures that the induced distributions $\propto \exp(\inprod{x}{\cdot} - \vhi)$ are heavy-tailed (because their log-densities cannot grow quickly), which means their variances are ``large.'' We also know that $\nabla^2 \psi$ is the covariance matrix of the induced distribution which means that $\nabla^2\psi$ should be lower-bounded. We formalize this using a version of the Cram\'er-Rao bound from \cite{chewipooladian2022caffarelli}. An older arXiv version of this paper contains a more elementary proof of this result inspired by differential privacy, achieving a worse constant of $\approx \frac 1 {12}$; the (optimal) improvement was suggested by Sam Power. Our converse proof is similar, and follows by applying the Brascamp-Lieb inequality \cite{BrascampL76}.

To prove our isoperimetric inequality, we draw inspiration from a similar bound shown in Lemma 35 of \cite{lee2018convergence}, but for a family of convex functions $\phi$ satisfying a non-standard condition that $\phi''$ was convex (which fortunately includes the $\log$ barrier function). Noticing that $-\log$ is self-concordant, we extend the \cite{lee2018convergence} result to hold for all self-concordant functions. Further we show by a direct calculation that the KL divergence between the induced distributions of two nearby points $x$ and $x'$ is essentially the LLT $\psi$ at one of the points, up to a linear term. This lets us use stability of the Hessian of self-concordance functions to demonstrate stability of nearby induced distributions, a key ingredient in proving conductance bounds by the machinery of \cite{DyerFK91}.

\paragraph{Non-Euclidean proximal sampling.} Given the results of Section~\ref{sec:llt}, establishing our main proximal sampling result Theorem~\ref{thm:mainllt} is fairly routine. Our algorithm consists of an ``outer loop'' and an ``inner loop'' for sampling from the $x$ marginal of \eqref{eq:density_intro} which is stated and analyzed in Section~\ref{ssec:xmarginal}. Our outer loop analysis is directly based on the mixing time-to-conductance reduction of \cite{lovasz1993random} and the technique of \cite{DyerFK91} to lower bound conductance, using facts from Section~\ref{sec:llt}. Our inner loop handling functions $F$ in \eqref{eq:density_intro} which are Lipschitz (or distributions over Lipschitz functions) is a small modification of a similar result in \cite{GLL22}. The only property we need of the LLT is strong convexity: this implies a rejection sampler terminates quickly via the concentration of Lipschitz functions under strongly logconcave distributions (in any norm) \cite{Ledoux99, BobkovL00}.

We do note there is a design decision to be made on how to define ``scaling up the LLT by $\frac 1 \eta$,'' unlike in the case of \eqref{eq:l2_joint} where using the induced density $\Nor(x, \eta^{-1} \id_d)$ is natural. Given $r$, a $1$-strongly convex function in $\normx{\cdot}$, and letting $r^*$ be its (smooth) Fenchel conjugate, two natural ways of defining a scaled up induced distribution at $x$ are to choose densities
\begin{equation}\label{eq:choice1}
\propto \exp\Par{\inprod{x}{y} - \eta r^*(y) - \psi(x)},
\end{equation}
or
\begin{equation}\label{eq:choice2}
	\propto \exp\Par{\frac 1 \eta \Par{\inprod{x}{y} -  r^*(y) - \psi(x)}}.
\end{equation}
The choice \eqref{eq:choice1} clearly results in $\psi$ which is $\Omega(\eta^{-1})$-strongly convex, rendering it suitable for our proximal sampling applications. It is not difficult to see that the second results in $\eta^{-1}\psi$ which is also $\Omega(\eta^{-1})$-strongly convex. More interestingly, plugging in $r = r^* = \half\norm{\cdot}_2^2$ makes \eqref{eq:l2_joint} agree with \eqref{eq:choice2} rather than \eqref{eq:choice1}. Unfortunately, the $\psi$ which results from \eqref{eq:choice2} is not self-concordant, as its Hessian scales with $\eta^{-1}$ and its third derivative with $\eta^{-2}$. Our choice to use \eqref{eq:choice1} has further implications, elaborated on next, but a deeper understanding of this discrepancy seems interesting.

\paragraph{Zeroth-order private convex optimization.} As outlined in Section~\ref{ssec:results}, the frameworks of \cite{GLL22, GopiLLST23} show that to use our proximal sampler for $\ell_p$ norm private convex optimization, it suffices to design an LLT which has small additive range. Perhaps surprisingly, we exploit the \emph{non-scale invariance} of LLT for this task: the LLT of $\eta \vhi$ does not behave like $\eta^{-1}$ times the LLT of $\vhi$.\footnote{On the other hand, the Fenchel conjugate of $\eta\vhi$ is $\eta^{-1}$ times the Fenchel conjugate of $\vhi$.} To see why this is helpful, consider the case when $\vhi = \half\norm{\cdot}_\infty^2$: then,
\[\psi(x) = \log\Par{\int \exp\Par{\inprod{x}{y} - \half\norm{y}_\infty^2} \dd y}.\]
Although one would hope $\psi(x)$ has additive range comparable to $\half\norm{x}_1^2$, the Fenchel conjugate of $\half\norm{x}_\infty^2$,
it is not hard to show that $\psi(e_1) - \psi(0) = \Omega(\sqrt{d})$; we give a proof in Appendix~\ref{app:sizebound}. Intuitively, the $\ell_\infty$ radius of a typical point $\sim \exp(-\half\norm{\cdot}_\infty^2)$ is about $\sqrt{d}$, and a constant fraction of points on the surface of this $\ell_\infty$ ball have inner product with $e_1$ of $\Omega(\sqrt d)$. This shows the additive range of $\psi$ on the $\ell_1$ ball is larger than $\half\norm{\cdot}_1^2$ by dimension-dependent factors.

We show that the non-scale invariance of \eqref{eq:choice1} is actually helpful in controlling additive ranges. Specifically, letting $\psi_\eta$ denote the LLT of $\eta \norm{x}_q^2$, we show the additive range of $\eta\psi_\eta$ (a $\approx 1$-strongly convex function) is $\approx \max(\eta, 1, \sqrt{d\eta})$. For sufficiently small $\eta$, this implies $\eta\psi_\eta$ is actually a much smaller regularizer than $\psi$; graciously, our differential privacy applications require $\eta \lesssim \frac 1 {d^2}$. We find it potentially useful to explore how generic this non-scale invariance of the LLT is.

\subsection{Prior work}\label{ssec:prior}

\paragraph{Non-Euclidean sampling.} 

A recurring issue that arises in bounding the convergence rate of non-Euclidean samplers is that na\"ive discretizations can result in significant error. As a result, most prior works either require strong assumptions or oracles for accurate discretization or adopt more sophisticated discretization methods that are difficult to analyze. For example, earlier in the introduction this was discussed for discretizations of MLD \cite{ZhangPFP20, Jiang21, AhnC21, LiTVW22}. Part of the intrinsic difficulty of bounding discretized MLD lies in third-order error terms emerging from non-Euclidean geometries, which are hard to control under standard assumptions.

Under structured settings different than, but related to, those in this paper, an interesting alternative sampling strategy is discretizing Riemannian Langevin or Hamiltonian dynamics. For example, \cite{gatmiry2022convergence} studied the Riemmanian Langevin dynamics assuming access to an oracle to sample from Brownian motion on a manifold, whose complexity heavily depends on the manifold. 
Further, the convergence rate of Riemannian Hamiltonian Monte Carlo (RHMC) in polytopes was studied in \cite{lee2018convergence}, and a discretized version was analyzed in \cite{kook2022condition}; the results apply to a limited family of distributions, and the convergence rate is fairly large. For RHMC to converge to the correct target distribution, sophisticated discretization methods such as Implicit Midpoint Method are necessary. Though efficient in practice, these methods are challenging to analyze theoretically. 

\paragraph{Proximal sampling.} A long line of works has studied the use of proximal methods in sampling (inspired by optimization). Several considered proximal Langevin
algorithms \cite{pereyra2016proximal,brosse2017sampling,bernton2018langevin,wibisono2019proximal}, which combine proximal methods and discretizations of Langevin dynamics. Further, \cite{mou2022efficient} proposed a sampler based on a proximal sampling oracle.  However, these algorithms required either stringent assumptions or a large mixing time. Recently, \cite{lee2021structured} proposed a new proximal sampler overcoming many of the assumptions and efficiency issues in prior methods. Several works have focused on generalizing \cite{lee2021structured} and applying it in different settings: \cite{ChenCSW22} proved convergence results using weaker assumptions than strong logconcavity. The framework has been used to obtain state-of-the-art samplers for various structured families, including smooth, composite, and finite-sum densities \cite{lee2021structured} as well as non-smooth densities \cite{GLL22, liang2022}.  

\paragraph{Log-Laplace transform.} The LLT is a powerful tool that emerges frequently in probability theory and convex geometry. Notably, \cite{bubeck2014entropic, Chewi21} showed that the Legendre-Fenchel dual of LLT of the uniform measure on a convex body in $\R^n$ is an $n$-self-concordant barrier, giving the first universal barrier for convex bodies with optimal self-concordance parameter. In \cite{ChenE22}, the LLT serves as one of the key ingredients of entropy conservation in localization schemes for sampling.  In addition, the LLT shows up in the solution to the entropic optimal transport problem, where a KL
divergence is added to regularize the optimal transport objective \cite{chewipooladian2022caffarelli}.  

\paragraph{Private convex optimization.} 
Differentially private convex optimization is one of the most extensively studied problems in the privacy literature and captures an increasing number of critical applications in various domains, including machine learning, statistics, and data analysis. 
There is a rich body of works on this topic \cite{CM08,CMS11,KST12,BST14,WYS17,BFTT19,FKT20},
which have mainly focused on the Euclidean geometry, e.g.\ assuming the $\ell_2$ diameter of the domain and $\ell_2$ norms of gradients are bounded.
Motivated by applications not captured by these assumptions, there has been growing interest in studying differentially private convex optimization in non-Euclidean geometries, as seen in \cite{ttz15,afkt21,bgn21,HLL+22,GopiLLST23}. Of particular relevance, \cite{GopiLLST23} develops an exponential mechanism based method attaining state-of-the-art excess risk bounds for $\ell_p$ and Schatten-$p$ norms, which are matched by our algorithms in Section~\ref{sec:apps}.

\subsection{Erratum since COLT 2023 version}\label{ssec:erratum}

An earlier version of this manuscript was presented at COLT 2023, which claimed an improvement to the present version of Theorem~\ref{thm:mainllt} that scaled linearly, rather than quadratically, in the quantity $\frac 1 {\eta\mu}$. This claim was erroneous, and we have weakened the claim to have the correct dependence, based on a modified isoperimetric inequality in Lemma~\ref{lem:iso} that scales with $m$ rather than $\sqrt{m}$. Correspondingly, the query complexities in our main application, Theorem~\ref{thm:lp}, are also weakened by roughly a quadratic factor. All other parts of our framework (the remainder of Section~\ref{sec:llt}, the rejection sampler in Section~\ref{ssec:xmarginal}, the size bounds in Section~\ref{ssec:lp} and implementations in Section~\ref{ssec:valueoracle}, and our various appendices) remain unchanged from the conference version. 

The key issue was an incorrect application of the localization lemma in Lemma~\ref{lem:iso}: for a self-concordant function $\phi$, and a constant $m \in (0, 1)$, $m\phi$ is not actually self-concordant (as $m^{\frac 3 2} \le m$). For potentials that are $m$-strongly convex in a norm, the typical scaling in Lemma~\ref{lem:iso} is indeed $\sqrt{m}$; unfortunately, because self-concordant functions behave ``locally quadratically'' only in a small range, and otherwise can have linear tails, the resulting isoperimetry instead scales as $m$. We believe our work nonetheless takes an important step by designing a fully-implementable framework for non-Euclidean sampling under minimal assumptions, and pose designing an improved framework that recovers the original scaling as $\frac 1 {\eta\mu}$ as an interesting and natural open problem.

\section{Preliminaries}
\label{sec:prelims}

\paragraph{General notation.} In Section~\ref{sec:intro} only, $\tO$, $\approx$, and $\lesssim$ hide logarithmic factors in problem parameters for expositional convenience. For $n \in \N$, $[n]$ refers to the naturals $1 \le i \le n$. We use $\xset$ to denote a compact convex subset of $\R^d$. For all $p \ge 1$ including $p = \infty$, we let $\norm{\cdot}_p$ applied to a vector argument denote the $\ell_p$ norm. We denote matrices in boldface and when $\norm{\cdot}_p$ is applied to a matrix argument it denotes the corresponding Schatten-$p$ norm ($\ell_p$ norm of the singular values).

For any $\xset \subset \R^d$ we let its indicator function (i.e.\ the function which is $1$ on $\xset$ and $0$ otherwise) be denoted $\1_{\xset}$. We will be concerned with optimizing functions $f: \xset \to \R$, and $\normx{\cdot}$ refers to a norm on $\xset$. We let $\xset^*$ be the dual space to $\xset$, and equip it with the dual norm $\normxd{y} \defeq \sup_{\normx{x} = 1} x^\top y$. We let $\Nor(\mu, \msig)$ be the Gaussian density of given mean and covariance. For a positive definite matrix $\mm \in \R^{d \times d}$, we denote the induced norm by $\norm{v}_{\mm} \defeq \sqrt{v^\top \mm v}$. When making asymptotic statements we will typically assume the dimension $d$ is at least a sufficently large constant, else we can pad and affect statements by at most constant factors.

\paragraph{Optimization.} In the following, fix $f: \xset \to \R$. We say $f$ is $G$-Lipschitz in $\normx{\cdot}$ if for all $x, x' \in \xset$, $|f(x) - f(x')| \le G\normx{x - x'}$. If $f$ is differentiable, we say it is $L$-smooth in $\normx{\cdot}$ if for all $x, x' \in \xset$, $\normxd{\nabla f(x) - \nabla f(x')} \le L \normx{x - x'}$. Taylor expanding then shows $f(x') \le f(x) + \inprod{\nabla f(x)}{x' - x} + \frac L 2\normx{x - x'}^2$. We say $f$ is $m$-relatively strongly convex in $\phi$ if $f - m\phi$ is convex. For $k$-times differentiable $f$, $\nabla^k f(x)[v_1, v_2, \ldots, v_k]$ denotes the corresponding $k^{\text{th}}$ order directional derivative at $f$. We say twice-differentiable $f$ is $m$-strongly convex in $\normx{\cdot}$ if for all $x \in \xset$, $v \in \R^d$, $\nabla^2 f(x)[v, v] \ge m\normx{v}^2$. We say convex $\phi: \R^d \to \R$ is self-concordant if \[\Abs{\nabla^3 \phi(x)[h,h,h]} \le 2\Par{\nabla^2 \phi(x)[h, h]}^{\frac 3 2}, \text{ for all } x, h \in \R^d.\]
A key fact we use about self-concordant functions is that their Hessians are stable under small distances, where the distance is measured in the Hessian norm: see Lemma~\ref{lem:selfcon_hess} for a formal statement.

\paragraph{Probability.} For a density $\pi$ supported on $\xset$, we let $\pi(S) \defeq \Pr_{x \sim \pi}[x \in S]$. For two densities $\mu, \pi$, we define their total variation distance by $\tvd{\mu}{\pi} \defeq \half \int |\mu(x) - \pi(x)|\dd x$ and (when the Radon-Nikodym derivative exists) their KL divergence by $\dkl(\mu \| \pi) \defeq \int \mu(x) \log \frac{\mu(x)}{\pi(x)} \dd x$. For $1 < \alpha < \infty$, we also define the $\alpha$-R\'enyi divergence between densities $\mu, \pi$ by
\[D_\alpha(\mu\|\pi) \defeq \frac 1 {\alpha - 1}\log \Par{\int \Par{\frac{\mu(x)}{\pi(x)}}^{\alpha}\pi(x) \dd x}.\]
We say density $\pi$ is logconcave (respectively, $m$-strongly logconcave in $\normx{\cdot}$) if $-\log \pi$ is convex (respectively, $m$-strongly convex in $\normx{\cdot}$). We similarly say $\pi$ is $m$-relatively strongly logconcave in $\phi$ if $-\log \pi$ is $m$-relatively strongly convex in $\phi$. If $\log \pi$ is affine, we say $\pi$ is logaffine.  We say a density $\pi_0$ is $\beta$-warm with respect to a density $\pi$ if for all $x$ in the support of $\pi$, $\frac {\dd \pi_0(x)}{\dd \pi(x)} \leq \beta$.

\paragraph{Log-Laplace transform.} We define the log-Laplace transform (LLT) of $\vhi: \R^d \to \R$ by
\[\psi(x) \defeq \log\Par{\int \exp\Par{\inprod{x}{y} - \vhi(y)} \dd y}.\]
When $\vhi, \psi$ are clear from context, we define the density
\begin{equation}\label{eq:npsidef}\Npsi_x(y) = \exp\Par{\inprod{x}{y} - \vhi(y) - \psi(x)}.\end{equation}
Note that the normalization constant is exactly given by $\psi(x)$ and hence $\Npsi_x$ is indeed a valid density.
We use $\propto$ to indicate proportionality, e.g.\ if $\mu$ is a density and we write $\mu \propto \exp(-f)$, we mean $\mu(x) = \frac{\exp(-f)}{Z}$ where $Z \defeq \int \exp(-f(x)) \dd x$ and the integration is over the support of $\mu$.

\paragraph{Riemannian geometry.} In Sections~\ref{sec:llt} and~\ref{sec:sampling} we will use geometry induced by the Hessian of a self-concordant, convex function $\phi: \R^d \to \R$. We summarize the important points here, and defer a more extended treatment to \cite{nesterov2002riemannian}. When $\phi$ is clear from context, we denote the norm $\norm{h}_x \defeq \norm{h}_{\nabla^2 \phi(x)}$. Throughout this discussion let $M \subseteq \R^d$ be a Riemannian manifold equipped with the local metric $\norm{\cdot}_x$. The induced Riemannian distance of a curve $c: [0, 1] \to M$ is defined as
\begin{align*}
L_\phi(c) \defeq \int_0^1 \norm{\frac {\dd} {\dd t} c(t)}_{c(t)} \dd t,
\end{align*}
where $\frac{\dd}{\dd t} c(t)$ is the velocity element of the curve in the tangent space at $c(t)$. For $x, y \in M$, we then define $d_\phi(x, y)$ to be the infimum of the length $L_\phi(c)$ over all curves $c$ such that $c(0) = x$ and $c(1) = y$. We will use the following two important properties of the Riemannian geometry over $M = \R^d$ induced by self-concordant, convex functions.

\begin{lemma}[\cite{nesterov2002riemannian}, Lemma 3.1]
	\label{lem:dist_const}
	Suppose $\phi: \R^d \to \R$ is convex and self-concordant. For $x,y \in \R^d$,  if $d_\phi(x,y)\leq \delta - \delta^2 < 1$ for some $\delta \in (0, 1)$, then $\norm{y-x}_{x}\leq \delta$.
\end{lemma}

\begin{lemma}[\cite{nemirovski2004interior}, Section 2.2.1]
	\label{lem:selfcon_hess}
	Suppose $\phi: \R^d \to \R$ is convex and self-concordant. For any $h, x \in \R^d$ such that $\norm{h}_{x} < 1$,
	$
		(1-\norm{h}_{x})^2  \nabla^2 \phi(x) \preceq \nabla^2 \phi(x+h) \preceq 	(1-\norm{h}_{x})^{-2 } \nabla^2 \phi(x).
	$
\end{lemma} 	%

\section{Properties of the LLT}
\label{sec:llt}

In this section, we collect a variety of facts about the log-Laplace transform which we will use to develop our sampling scheme in Section~\ref{sec:sampling}. We begin by proving basic facts about the LLT in Section~\ref{ssec:basics}. We then use them to derive isoperimetric properties of induced distributions in Section~\ref{ssec:isoperimetry} and total variation bounds in Section~\ref{ssec:conductance}. Throughout this section we will fix a convex function $\vhi: \R^d \to \R$, and let $\psi: \R^d \to \R$ be its LLT. We will also follow the notation \eqref{eq:npsidef}.

\subsection{Basic properties and duality}\label{ssec:basics}

The log-Laplace transform $\psi$ at $x$ is the cumulant-generating function of the distribution $\Npsi_x$, which means that $\psi$ is infinitely-differentiable and that $\nabla^k \psi$ is the $k^{\text{th}}$ cumulant tensor of $\Npsi_x$. We will only use the first three derivatives of $\psi$ which we compute below for completeness.

\begin{lemma}[LLT derivatives]
	\label{lem:cumulant}
	For any $x, h \in \R^d$, we have
	\begin{align*}
	\nabla \psi(x) &= \mu(\Npsi_x) \defeq \E_{y \sim \Npsi_x}[y], \\
	\nabla^2 \psi(x) &= \mathrm{Cov}(\Npsi_x) \defeq \E_{y \sim \Npsi_x}\Brack{(y - \mu(\Npsi_x))(y - \mu(\Npsi_x))^\top}, \\
	\nabla^3 \psi(x)[h, h, h] &= \E_{y \sim \Npsi_x}\Brack{\inprod{y - \mu(\Npsi_x)}{h}^3}. 
	\end{align*}
\end{lemma}
\begin{proof}
	For any $x\in \R^d$, a straightforward calculation shows that 
	\begin{align*}
		\nabla \psi(x)& = \nabla \Par{\log \int\exp\Par{ \inprod{x}{y} -\varphi(y)}\dd y} 
		= \frac {\int \exp\Par{ \inprod{x}{y} - \varphi(y)}y \dd y}{ \int \exp\Par{ \inprod{x}{y} - \varphi(y)} \dd y} = \mu (\Npsi_x ).
	\end{align*}
	Further,
	\begin{align*}
		\nabla^2 \psi(x)& = \nabla \Par{ \frac {\int \exp\Par{ \inprod{x}{y} - \varphi(y)}y \dd y}{ \int \exp\Par{ \inprod{x}{y} - \varphi(y)} \dd y}} \\
		& = \frac{{ \int \exp\Par{ \inprod{x}{y} - \varphi(y)}yy^\top \dd y}}{{ \int \exp\Par{ \inprod{x}{y} - \varphi(y)}\dd y}} - \frac{\Par{ \int \exp\Par{ \inprod{x}{y} - \varphi(y)}y \dd y}\Par{ \int \exp\Par{ \inprod{x}{y} - \varphi(y)}y \dd y}^\top}{\Par{ \int\exp\Par{ \inprod{x}{y} -\varphi(y)}\dd y}^2}.
	\end{align*}
	Finally, 
	\begin{align*}
		\nabla^3 \psi(x)[h,h,h] 
		&= h^\top \nabla \Par{\frac{{ \int \exp\Par{ \inprod{x}{y} - \varphi(y)}\Par{y^\top h}^2 \dd y}}{{ \int\exp\Par{ \inprod{x}{y} - \varphi(y)} \dd y}} - \frac{\Par{ \int\exp\Par{ \inprod{x}{y} - \varphi(y)}y^\top h \dd y}^2}{\Par{ \int \exp\Par{ \inprod{x}{y} -\varphi(y)}\dd y}^2}}\\
		&= \frac {{ \int \exp\Par{ \inprod{x}{y} - \varphi(y)} \Par{y^\top h}^3 \dd y}}{{ \int\exp\Par{ \inprod{x}{y} - \varphi(y)} \dd y}} +2\Par{ \frac {\int \exp\Par{ \inprod{x}{y} - \varphi(y)}y^\top h\dd y}{{ \int \exp\Par{ \inprod{x}{y} - \varphi(y)} \dd y}}}^3 \\
		& -   \frac {3{ \int \exp\Par{ \inprod{x}{y} - \varphi(y)} \Par{y^\top h}^2 \dd y}{ \int \exp\Par{ \inprod{x}{y} - \varphi(y)}{y^\top h} \dd y}}{\Par{ \int\exp\Par{ \inprod{x}{y} - \varphi(y)} \dd y}^2}.
	\end{align*}
\end{proof}

By using a fact on one-dimensional logconcave distributions in \cite{bubeck2014entropic}, this implies the following.

\begin{lemma}[Self-concordance]\label{lem:psi_sc}
If $\psi$ is the LLT of a convex function, it is self-concordant.
\end{lemma}
\begin{proof}
	By the definition of self-concordance and Lemma~\ref{lem:cumulant}, it suffices to show for any $h\in \R^d$,
	\begin{equation}\label{eq:psi_sc}\E_{y\sim \Npsi_x}\Brack{\inprod{y - \mu(\Npsi_x)}{h} }^3 \leq 2 \Par{\E_{y \sim \Npsi_x}\Brack{\inprod{y - \mu(\Npsi_x)}{h}^2}}^{\frac 3 2}.
	\end{equation}
	We then note that the random variable $\langle y - \mu(\Npsi_x), h\rangle$ for $y \sim \Npsi_x$ follows a logconcave distribution because affine transformations preserve logconcavity. Finally Lemma 2 of \cite{bubeck2014entropic} implies \eqref{eq:psi_sc} holds.
\end{proof}

Next, we prove that a form of strong convexity-smoothness duality (and its converse) holds with respect to $\vhi$ and $\psi$, analogous to the type of duality satisfied by Fenchel conjugates \cite{KakadeST09}.

\begin{lemma}[Strong convexity-smoothness duality]\label{lem:scllt}
	If $\vhi: \R^d \to \R$ is $L$-smooth with respect to $\norm{\cdot}_*$, then $\psi: \R^d \to \R$ is $\frac 1 L$-strongly convex with respect to $\norm{\cdot}$.
\end{lemma}
\begin{proof}
	By definition of strong convexity it suffices to prove for any $x, v \in \R^d$, $v^\top \nabla^2\psi(x)v \ge \frac 1 L \norm{v}^2$. Without loss of generality, by scale invariance we can assume $\norm{v} = 1$. Let $Y = \inprod{y}{v}$, where $y \sim \Npsi_x$. By Lemma~\ref{lem:cumulant}, $\nabla^2 \psi(x) = \Cov(\Npsi_x)$, so it suffices to prove that
	\[\Var(Y) = \E_{y\sim \Npsi_x}\Brack{\inprod{y-\mu(\Npsi_x)}{v}^2}  \geq \frac{1}{L}.\]	
	Letting $\mm \defeq \E_{y \sim \Npsi_x} \nabla^2 \vhi(y)$, we first observe
	\begin{align*}
		\frac L {2} v^\top\mm^{-1} v =  \max_{u \in \R^d} \inprod{u}{v} - \frac 1 {2L} u^\top \mm u \ge \max_{u \in \R^d} \inprod{u}{v} - \half \norm{u}_*^2 = \half \norm{v}^2.
	\end{align*}
	In the only inequality, we used that $u^\top \mm u = \E_{y \sim \Npsi_x} u^\top \nabla^2 \vhi(y) u \le L\norm{u}_*^2$ by smoothness of $\vhi$, and the last equality follows by optimizing over $\norm{u}_*$. This shows $v^\top \mm^{-1} v \ge \frac 1 L$. The Cram\'er-Rao inequality (see Lemma 2, \cite{chewipooladian2022caffarelli}) then implies
	\[\Var(Y) \ge v^\top \mm^{-1} v \ge \frac 1 L,\]
	since the Hessian of $-\log \Npsi_x$ at any $x \in \R^d$ is $\nabla^2 \vhi$.
\end{proof}

\begin{lemma}[Smoothness-strong convexity duality]\label{lem:smoothllt}
If $\vhi: \R^d \to \R$ is $\frac 1 L$-strongly convex with respect to $\norm{\cdot}_*$, then $\psi: \R^d \to \R$ is $L$-smooth with respect to $\norm{\cdot}$.
\end{lemma}
\begin{proof}
Let $v, x \in \R^d$ and assume $\norm{v} = 1$. As in Lemma~\ref{lem:scllt}, defining $Y = \inprod{y}{v}$ for $y \sim \Npsi_x$, we have $v^\top \nabla^2 \psi(x) v = \Var(Y)$, and want to show $\Var(Y) \le L$. First note that for any $y \in \R^d$,
\begin{align*}
	\frac 1 {2L} v^\top\Par{\nabla^2 \vhi(y)}^{-1} v =  \max_{u \in \R^d} \inprod{u}{v} - \frac L 2 u^\top \nabla^2 \vhi(y) u \le \max_{u \in \R^d} \inprod{u}{v} - \half \norm{u}_*^2 = \half \norm{v}^2.
\end{align*}
The first inequality used strong convexity of $\vhi$ and the last equality follows by optimizing over $\norm{u}_*$. This shows $v^\top (\nabla^2 \vhi(y))^{-1} v \le L$ for all $y$. 
The Brascamp-Lieb inequality \cite{BrascampL76} then implies
\[\Var(Y) \le \E_{y \sim \Npsi_x}\Brack{v^\top\Par{\nabla^2 \vhi(y)}^{-1} v} \le L,\]
since the Hessian of $-\log \Npsi_x$ at any $x \in \R^d$ is $\nabla^2\vhi$. 
\end{proof}

\subsection{Isoperimetry}\label{ssec:isoperimetry}

In this section we prove Lemma~\ref{lem:iso}, an isoperimetric inequality for densities which are relatively strongly logconcave with respect to an appropriate LLT. A stronger variant of Lemma~\ref{lem:iso} appeared in the conference version of this manuscript, and this version corrects the claim. We note that a similar (and somewhat more general) variant of the current Lemma~\ref{lem:iso} has also appeared earlier as Lemma 18 in \cite{SrinivasanWW25} (see also Lemma 23 in the same paper, a variant of Lemma~\ref{lem:iso_onedim}).

The only LLT property we use in this section is Lemma~\ref{lem:psi_sc}, i.e.\ self-concordance, via the following generic fact which generalizes Lemma 35 of \cite{lee2018convergence}.

\begin{lemma}
	\label{lem:iso_onedim}
	Suppose $\phi: \R \to \R$ is convex and self-concordant. For any $x \in \R$,
	\[\frac{\exp(-\phi(x))}{\sqrt{\phi''(x)}} \ge \frac 1 {12} \min\Brace{\int_{-\infty}^x \exp(-\phi(t))\dd t, \int_{x}^\infty \exp(-\phi(t))\dd t}. \]
\end{lemma}
\begin{proof}
	Assume $\phi'(x) \ge 0$ (the other case will follow analogously by bounding the integral on $(-\infty, x]$). Define $r \defeq x + \frac 1 {4\sqrt{\phi''(x)}}$. By self-concordance (Lemma~\ref{lem:selfcon_hess}), for all $t \in [x, r]$,
	\[ \half \phi''(x) \le \phi''(t) \le 2\phi''(x).\]
	Hence, we have for all $t \in [x, r]$, since $\phi'(x) \ge 0$,
	\begin{equation}\label{eq:quadbound}
		\phi(t) = \phi(x) + \phi'(x)(t - x) + \int_x^t (t-s) \phi''(s) \dd s \ge \phi(x) + \frac 1 4 (t-x)^2 \phi''(x).
	\end{equation}
	We use \eqref{eq:quadbound} to bound the integral on $[x, r]$:
	\begin{equation}\label{eq:insideint}
		\begin{aligned}
			\int_x^r \exp(-\phi(t)) \dd t &\le \exp(-\phi(x)) \int_x^r \exp\Par{-\frac 1 4 (t-x)^2 \phi''(x)} \dd t \\
			&\le \exp(-\phi(x)) \int_{-\infty}^{\infty} \exp\Par{-\frac 1 4 (t-x)^2 \phi''(x)} \dd t = 2\sqrt{\pi} \cdot \frac{\exp(-\phi(x))}{\sqrt{\phi''(x)}}.
		\end{aligned}
	\end{equation}
	Next, to bound the integral on $[r, \infty)$, we first observe
	\begin{equation}\label{eq:linbound}\phi'(r) \ge \phi'(x) + \int_x^r \phi''(t) \dd t \ge \half \int_x^r \phi''(x) \dd t \ge \frac 1 8 \sqrt{\phi''(x)}. \end{equation}
	Hence, by convexity from $r$,
	\begin{equation}\label{eq:outsideint}
		\begin{aligned}
			\int_r^\infty \exp(-\phi(t)) \dd t &\le \int_r^\infty \exp\Par{-\phi(r) - \phi'(r) (t - r) } \dd t \\
			&\le \exp(-\phi(x)) \int_r^\infty \exp\Par{-\frac 1 8 \sqrt{\phi''(x)} (t - r)}\dd t = 8 \cdot \frac{\exp(-\phi(x))}{\sqrt{\phi''(x)}}.\end{aligned}
	\end{equation}
	We used $\phi(r) \ge \phi(x)$ by convexity and $\phi'(x) \ge 0$. Combining \eqref{eq:insideint} and \eqref{eq:outsideint} yields the claim.
\end{proof}

We will use a generalization of Lemma~\ref{lem:iso_onedim} to scalings of self-concordant functions. Note that we are only concerned with scalings $m \in (0, 1]$ in the following generalized statement, as for $m \ge 1$, $\phi$ being self-concordant implies that $m\phi$ is also self-concordant (since $m^{\frac 3 2} \ge m$).

\begin{corollary}\label{cor:scale_iso_onedim}
In the setting of Lemma~\ref{lem:iso_onedim}, let $m \in (0, 1]$. For any $x \in \R$,
\[\frac{\exp(-m\phi(x))}{\sqrt{\phi''(x)}} \ge \frac m {12} \min\Brace{\int_{-\infty}^x \exp(-m\phi(t)) \dd t, \int_x^\infty \exp(-m\phi(t)) \dd t}.\]
\end{corollary}
\begin{proof}
As in Lemma~\ref{lem:iso_onedim}, we only handle $\phi'(x) \ge 0$, and let $r \defeq x + \frac 1 {4\sqrt{\phi''(x)}}$. We again bound:
\begin{align*}
\int_x^r \exp(-m\phi(t))\dd t &\le \exp(-m\phi(x)) \int_{-\infty}^\infty \exp\Par{-\frac m 4 (t - x)^2\phi''(x)} \dd t \\
&= 2\sqrt{\pi} \cdot \frac{\exp(-m\phi(x))}{\sqrt{m\phi''(x)}} \le 2\sqrt{\pi} \cdot \frac{\exp(-m\phi(x))}{m\sqrt{\phi''(x)}},
\end{align*}
where we used \eqref{eq:quadbound}, and
\begin{align*}
\int_r^\infty \exp(-m\phi(t))\dd t \le \exp(-m\phi(x)) \int_r^\infty \exp\Par{-\frac m 8 \sqrt{\phi''(x)}(t-r)} \dd t = 8 \cdot \frac{\exp(-m\phi(x))}{m\sqrt{\phi''(x)}},
\end{align*}
where we used \eqref{eq:linbound}. Combining these two displays implies the claim.
\end{proof}

Next, we reduce the problem of proving isoperimetry for relatively strongly logconcave densities to the same problem in one dimension (captured via Lemma~\ref{lem:iso_onedim}), via the localization lemma.

\begin{lemma}[Modification of the localization lemma, \cite{kannan1995isoperimetric}, Theorem 2.7]
	\label{lem:local}
	Let $f_1, f_2, f_3, f_4$ be four nonnegative functions on $\R^d$ such that $f_1$ and $f_2$ are upper semicontinuous and $f_3$ and $f_4$ are lower semicontinuous, let $c_1, c_2 >0$, and let $\phi: \R^d \to \R$ be convex. Then, the following are equivalent:
	\begin{itemize}
		\item For every density $\pi: \R^d \to \R$ which is $1$-relatively strongly logconcave in $\phi$,
		\begin{equation*}
			\Par{\int f_1(x) \pi(x) \dd x} ^{c_1} \Par{\int f_2(x) \pi(x) \dd x} ^{c_2} \leq \Par{\int f_3(x) \pi(x) \dd x} ^{c_1} \Par{\int f_4(x) \pi(x) \dd x} ^{c_2}.  
		\end{equation*}
		\item For every $a, b \in \R^d$ and $\gamma \in \R$, 
		\begin{align*}
			\Par{\int_0^1 f_1((1-t)a+tb)e^{\gamma t - \phi((1-t)a+tb)} \dd t}^{c_1}	\Par{\int_0^1 f_2((1-t)a+tb)e^{\gamma t - \phi((1-t)a+tb)} \dd t}^{c_2} \\
			\leq 	\Par{\int_0^1 f_3((1-t)a+tb)e^{\gamma t - \phi((1-t)a+tb)}\dd t}^{c_1}	\Par{\int_0^1 f_4((1-t)a+tb)e^{\gamma t - \phi((1-t)a+tb)}\dd t}^{c_2}.
		\end{align*}
	\end{itemize}
\end{lemma}                     
\begin{proof}
The proof follows identically to the case where $\phi = 0$, which was proven in \cite{lovasz1993random, kannan1995isoperimetric} via a bisection argument (see Lemma 2.5, \cite{lovasz1993random}). The only fact the bisection argument relies on is that restricting logconcave densities to subsets of $\R^d$ preserves logconcavity, which remains true for densities which are relatively strongly logconcave with respect to a given convex function. For a more formal treatment of this generalized bisection argument, see Lemma 1 of \cite{GopiLLST23}. Finally the change on the continuity assumptions on the $\{f_i\}_{i \in [4]}$ follows by Remark 2.3 of \cite{kannan1995isoperimetric}.
\end{proof}
Finally, we combine these tools to prove the main result of this section.

\begin{lemma}[Self-concordant isoperimetry]\label{lem:iso}
	Let $\phi: \R^d \rightarrow \R$ be convex and self-concordant, and let $f: \R^d \to \R$ be $m$-relatively strongly convex in $\phi$ for $m \in (0, 1]$. For any partition $S_1, S_2, S_3$ of $\R^d$, 
	\begin{equation*}
		\frac{\int_{S_3}\exp\Par{-f(x)} \dd x}{\min \Brace{\int_{S_1}\exp(-f(x))\dd x, \int_{S_2}\exp(-f(x))\dd x}} = \Omega\Par{m \cdot d_{\phi}(S_1, S_2)},
	\end{equation*}
	where $d_{\phi}(S_1, S_2) =\min_{x\in S_1, y\in S_2} d_{\phi}(x,y) $.
\end{lemma}
\begin{proof}
We first show that without loss of generality, we can assume \begin{equation}\label{eq:bigset}\max_{i\in\{1,2\}}  \frac {{\int_{S_i}\exp(-f(x))\dd x}}{{\int\exp(- f(x))\dd x}} = \Omega(1).\end{equation}
	To see this, let $S_1^\star, S_2^\star$ and $S_3^\star$ be the partition that achieves the minimum of \[\beta(S_1,S_2,S_3) \defeq \frac{\int_{S_3}\exp\Par{- f(x)} \dd x}{ d_{\phi}(S_1,S_2)\min \Brace{\int_{S_1}\exp(- f(x))\dd x, \int_{S_2}\exp(- f(x))\dd x}}.\]
	Let $\delta = d_{\phi}(S_1^\star,S_2^\star)$. For any $z \in S_3^\star$, let $x \in S_1^\star$ minimize $d_\phi(x, z)$ and let $y \in S_2^\star$ minimize $d_\phi(y, z)$. By the triangle inequality we have
	\[d_\phi(x, z) + d_\phi(y, z) \ge \delta\]
	and hence $\max(d_\phi(x, z), d_\phi(y, z) ) \ge \frac \delta 2$. Consequently we can partition $S_3^\star$ into $S_3'$ and $S_3''$ such that $ d_{\phi}(S_1^\star,S_3') \ge \frac \delta 2 $ and $d_{\phi}(S_2^\star, S_3'') \ge \frac \delta 2 $ by placing each $z$ into an appropriate set. Moreover, we can assume without loss of generality that the measure of $S_3'$ according to $\exp(-f)$ is at most half the measure of $S_3$ (else this is true for $S_3''$). This implies
	$$\frac{\int_{S_3'}\exp\Par{- f(x)} \dd x}{ \frac \delta 2 \min \Brace{\int_{S_1^\star}\exp(- f(x))\dd x, \int_{S_2^\star}\exp(- f(x))\dd x}}\leq \beta(S_1^\star, S_2^\star, S_3^\star).$$
	Thus, $\beta(S_1^\star \cup S_3'', S_2^\star,  S_3') \leq \beta(S_1^\star , S_2^\star,  S_3^\star)$, proving \eqref{eq:bigset} (else we may halve the measure of $S_3$ by replacing it with $S_3'$, guaranteeing \eqref{eq:bigset} holds). 
	
   Given \eqref{eq:bigset}, it suffices to show that there is a constant $C$ with
	\begin{align*}
		C m \cdot d_{\phi}(S_1,S_2)\int \exp(-f(x))\1_{S_1}(x) \dd x \int \exp(-f(x))\1_{S_2}(x) \dd x \\
		\leq \int \exp(-f(x))\dd x \int \exp(-f(x))\1_{S_3}(x)\dd x.
	\end{align*}
	Using the localization lemma (Lemma~\ref{lem:local}), letting $f_i = \1_{S_i}$ for $i \in [3]$ and $f_4 = (Cm \cdot d_\phi(S_1, S_2))^{-1}$,\footnote{Without loss of generality we can assume $S_1$ and $S_2$ are closed (implying $S_3$ is open) by taking their closures. This implies $f_1$, $f_2$ are upper semicontinuous and $f_3$, $f_4$ are lower semicontinuous.} and overloading $\phi \gets m\phi$, it suffices to prove for every $a, b \in \R^d$ and $\gamma \in \R$,
	\begin{align*}
		&Cm \cdot d_{\phi}(S_1,S_2) \int_0^1 \exp\Par{\gamma t-m\phi((1-t)a+tb)}\1_{S_1}((1-t)a+tb)\dd t \\
		\cdot &  \int_0^1 \exp\Par{\gamma t -m\phi((1-t)a+tb)}\1_{S_2}((1-t)a+tb) \dd t  \\
		\leq & \int_0^1 \exp\Par{\gamma t-m\phi((1-t)a+tb)} \dd t \int_0^1 \exp\Par{\gamma t-m\phi((1-t)a+tb)}\1_{S_3}((1-t)a+tb)\dd t.
	\end{align*} 
	Redefine $\phi(t) \gets \phi((1 - t)a + tb) - \frac{\gamma t} m$ for $t \in \R$, which is a one-dimensional self-concordant function, and redefine $S_i \gets \{t \mid (1 - t)a + tb \in S_i\}$ for $i \in [3]$, such that each $S_i$ is a union of intervals. It is straightforward to check that the distance $d_\phi(S_1, S_2)$ only increases under this transformation, because it can only take fewer paths (as they must now be along the line between $a$ and $b$).
	
	 So, it suffices to consider the special one-dimensional case with $\gamma = 0$, where $d_{\phi}(x,y) = \int_x^y \sqrt{\phi''(t)} \dd t$. We next note that it suffices to consider the case when $S_3$ is a single interval, i.e.\ for any $a \le a' \le b' \le b$, we have $S_1 = [a, a']$, $S_2 = [b', b]$, $S_3 = [a', b']$, and wish to show for some constant $C$
	\begin{equation}
		\label{eq:iso_onedim}
		\frac{\int_{a'}^{b'}\exp(-m\phi(t))\dd t}{\int_{a'}^{b'}\sqrt {\phi''(t)}dt} \geq Cm \cdot  \frac{\int_a^{a'}\exp(-m\phi(t))\dd t\int_{b'}^b\exp(-m\phi(t))\dd t}{\int_{a}^{b}\exp(-m\phi(t))\dd t}.
	\end{equation}
	When $S_3$ has multiple intervals, by Theorem 2.6 in \cite{lovasz1993random}, we show \eqref{eq:iso_onedim} for each interval in $S_3$ and its adjacent segments in $S_1$ and $S_2$, and sum over all such inequalities. 
	Finally, by Corollary~\ref{cor:scale_iso_onedim}, when $\phi$ is convex and self-concordant, we have for any $x\in[a,b]$,
	\begin{equation*}
		\frac{\exp(-m\phi(x))}{\sqrt{\phi''(x)}} \geq \frac m {12} \min\Par{\int_{a}^x \exp(-m\phi(t))\dd t, \int_{x}^b \exp(-m\phi(t))\dd t}
	\end{equation*}
	which combined with $	\frac{\int_{a'}^{b'} \exp(-m\phi(t))\dd t}{\int_{a'}^{b'}\sqrt{\phi''(t)}\dd t} \geq \min_{x\in[a',b']} \frac{\exp(-m\phi(x))}{\sqrt{\phi''(x)}}$ shows \eqref{eq:iso_onedim}. 
\end{proof}

\subsection{Total variation bounds}\label{ssec:conductance}

In this section, we provide a bound on the total variation distance of induced distributions $\Npsi_x$ and $\Npsi_{x'}$, when $x$ and $x'$ are close in the Riemannian distance given by $\psi$.

\begin{lemma}[TV distance between $\Npsi_x$ and $\Npsi_{x'}$]\label{lem:tvdist}
	For any $x,x'\in\R^d$ such that $d_{\psi}(x,x')\leq \frac 1 4$, 
	\[\tvd{\Npsi_x}{\Npsi_{x'}}\leq \frac 1 2.\]
\end{lemma}
\begin{proof}
Let $h = x' - x$ and note that the KL divergence between $\Npsi_x$ and $\Npsi_{x'}$ may be rewritten as
\begin{align*}D_{\text{KL}}\Par{\Npsi_x \| \Npsi_{x'}} &= \E_{y \sim \Npsi_x}\Brack{\log \frac{\dd\Npsi_x}{\dd\Npsi_{x'}}(y)} = \E_{y \sim \Npsi_x}\Brack{\psi(x') - \psi(x) - \inprod{h}{y}} \\
&= \psi(x') - \psi(x) - \inprod{h}{\nabla \psi(x)}.
\end{align*}
In the last equation, we used Lemma~\ref{lem:cumulant}. We recognize that the KL divergence is the Bregman divergence (first-order Taylor approximation) in $\psi$, and hence letting $x_t = x + th$ for $t \in [0, 1]$ such that $x_0 = x$ and $x_1 = x'$, we continue bounding
\begin{align*}
D_{\text{KL}}\Par{\Npsi_x \| \Npsi_{x'}} &= \int_0^1(1 - t) \nabla^2 \psi(x_t)[h, h]\dd t \\
&\le \int_0^14(1 - t) \nabla^2 \psi(x)[h, h] \dd t \le \half.
\end{align*}
The first inequality used that when $d_\psi(x, x') \le \frac 1 4$, Lemma~\ref{lem:dist_const} shows $\norm{x_t - x}_x \le \norm{x' - x}_x \le \half$, so Lemma~\ref{lem:selfcon_hess} gives $\nabla^2 \psi(x_t) \preceq 4\nabla^2 \psi(x)$; the second used $\norm{h}_x \le \half$. Finally by Pinsker's inequality,
	\begin{equation*}
		\tvd{\Npsi_x}{\Npsi_{x'}}\leq \sqrt{\half D_{\text{KL}}(\Npsi_x \| \Npsi_{x'}) } \leq \frac 1 2.
	\end{equation*}
	
\end{proof} 	%

\section{Proximal LLT sampler}
\label{sec:sampling}

In this section, we study a sampling problem in the following setting, assumed throughout.
\begin{problem}\label{prob:mainsetting}
For $D, G, \eta> 0$, let $\xset \subset \R^d$ be compact and convex, with diameter in a norm $\normx{\cdot}$ at most $D$. 
Let $F: \xset \to \R$ have the stochastic form $F(x) \defeq \E_{i \sim \ind} \Brack{f_i(x)}$, for a distribution $\ind$ over (a possibly infinite) family of indices $i$, such that each $f_i: \xset \to \R$ is convex and $G$-Lipschitz in $\normx{\cdot}$. 
Finally, let $\vhi: \R^d \to \R$ be convex and $\eta$-smooth in the dual norm $\normxd{\cdot}$. Given $\mu > 0$, and letting $\psi: \R^d \to \R$ be the LLT of $\vhi$, the goal is to sample from the density $\pi$ satisfying
\begin{equation}\label{eq:maindensity}
	\dd\pi(x) \propto  \exp\Par{-F(x)-\eta \mu \psi(x)} \1_{\xset}(x) \dd x.
\end{equation}
\end{problem}
Note that by Lemma~\ref{lem:scllt}, $\eta\mu\psi$ is $\mu$-strongly convex in $\normx{\cdot}$. Letting $z = (x, y)$ denote a variable on $\xset \times \R^d$, it is convenient for us to define the extended density on the joint space of $z$:
\begin{equation}\label{eq:jointdensity}
	\dd\hpi(z) \propto \exp\Par{-F(x) - \eta\mu\psi(x) + \Par{\inprod{x}{y} - \psi(x) - \vhi(y)}} \1_{\xset}(x) \dd z.
\end{equation}
Our sampling framework for \eqref{eq:maindensity} generalizes an approach pioneered by \cite{lee2021structured}, and is stated in the following Algorithm~\ref{alg:alternatesample}. The algorithm simply alternately samples from each marginal of \eqref{eq:jointdensity}. Before stating it, we define the following notation for conditional densities throughout the section:
\begin{equation}\label{eq:marginaldef}
\begin{aligned}
\dd \pi_x(y) &= \exp\Par{\inprod{x}{y} - \psi(x) - \vhi(y)} \dd y \text{ for all } x \in \xset, \\
\dd \pi_y(x) &\propto \exp\Par{-F(x) - (1 + \eta\mu) \psi(x) + \inprod{x}{y}} \1_{\xset}(x) \dd x \text{ for all } y \in \R^d.
\end{aligned}
\end{equation}
In particular, we observe that $\dd \pi_x(y) = \dd \hpi(\cdot \mid x)$ and $\dd\pi_y(x) = \dd \hpi(\cdot \mid y)$.
\begin{algorithm}[ht!]\caption{$\AlternateSample(\xset, F, \vhi, T, \mu, x_0)$}
	\label{alg:alternatesample}
	\textbf{Input:} $\xset, F, \vhi$ in the setting of Problem~\ref{prob:mainsetting}, $T \in \N$, $\mu > 0$, $x_0 \in \xset$.
	\begin{algorithmic}[1]
		\For{$k \in [T]$}
		\State Sample $y_k \sim \pi_{x_{k - 1}}$.\label{line:ysample}
		\State Sample $x_k \sim \pi_{y_k}$.\label{line:xsample}
		\EndFor
		\State \Return $x_T$
	\end{algorithmic}
\end{algorithm}

Correctness of Algorithm~\ref{alg:alternatesample} for sampling from \eqref{eq:jointdensity} builds upon the following basic facts.

\begin{lemma}\label{lem:alternatesample}
The total $x$-marginal of $\hpi$ in \eqref{eq:jointdensity} is $\pi$ in \eqref{eq:maindensity}. Furthermore, the stationary distribution of Algorithm~\ref{alg:alternatesample} is $\hpi$, and the induced Markov chains in Algorithm~\ref{alg:alternatesample} restricted to either $\{x_k\}_{0 \le k \le T}$ (a Markov chain on $\xset$) or $\{y_k\}_{k \in [T]}$ (a Markov chain on $\R^d$) are both reversible.
\end{lemma}
\begin{proof}
The first conclusion is a direct calculation, and the remainder is Lemma 1 in \cite{lee2021structured}.
\end{proof}

In Section~\ref{ssec:xmarginal} we develop a subroutine based on rejection sampling for implementing Line~\ref{line:xsample} of Algorithm~\ref{alg:alternatesample}, extending \cite{GLL22}. We then give our complete analysis of Algorithm~\ref{alg:alternatesample} in Section~\ref{ssec:mixing}.

\subsection{Sampling from the $x$-conditional distribution}\label{ssec:xmarginal}
Throughout this section, we assume the setting in Problem~\ref{prob:mainsetting}, and fix some $y \in \R^d$. We provide a sampler for the marginal density $\pi_y$ (following notation \eqref{eq:marginaldef}), and denote the component of the density independent of $F$ by $\gamma_y$, i.e.\
\begin{equation}\label{eq:gammadef}\dd \gamma_y(x) \propto \exp\Par{-\eta\mu\psi(x) - \Par{\psi(x) - \inprod x y}} \1_{\xset}(x) \dd x.
\end{equation}
By Lemma~\ref{lem:scllt}, $\gamma_y$ (and hence $\pi_y$) is $\frac 1 {\eta}$-strongly logconcave in $\normx{\cdot}$. Our rejection sampler leverages this fact and the stochastic nature of $F$ to build a rejection sampling scheme similarly to \cite{GLL22}. For completeness, we state our Algorithm~\ref{alg:inner} below, and provide the details of its analysis here.

\begin{algorithm}[ht!]\caption{$\Inner(y, \delta, \xset, F, \vhi, \mu)$}
	\label{alg:inner}
	\textbf{Input:} $\delta \in (0, \half)$, $y \in \R^d$, $\xset, F, \vhi$ in the setting of Problem~\ref{prob:mainsetting} for $\frac 1 \eta \ge 10^4 G^2 \log \frac 1 \delta$ \\
	\textbf{Output:} Sample within total variation distance $\delta$ of 
	\[\dd \pi_y(x) \propto \exp\Par{-F(x) - \eta\mu\psi(x) - \Par{\psi(x) - \inprod{x}{y}}} \1_{x \in \xset} \dd x.\]
	\begin{algorithmic}[1]
		\State $u \gets 1, \rho \gets 1$
  \While{$u > \half\rho$}
  \State Sample $x_1, x_2 \sim \gamma_y$ defined in \eqref{eq:gammadef} independently
		\State $\rho \gets 1$, $u \sim_{\textup{unif.}} [0, 1]$
		\State Draw $a \in \N$ such that for all $b \in \N$, $\Pr[a \ge b] = \frac 1 {b!}$
		\For{$b \in [a]$}
		\State Draw $\jia \sim \ind$ for $i \in [b]$\label{line:jiadraw}
		\State $\rho \gets \rho + \prod_{i \in [b]} (f_{\jia}(x_2) - f_{\jia}(x_1))$
		\EndFor
		\EndWhile
		\State \Return $x_1$
	\end{algorithmic}
\end{algorithm}

In order to analyze Algorithm~\ref{alg:inner},  we first state a general result about concentration of Lipschitz functions with respect to a strongly logconcave measure, in general norms. The following is a direct adaptation of standard results on log-Sobolev inequalities contained in \cite{Ledoux99, BobkovL00}.

\begin{lemma}[\cite{Ledoux99}, Section 2.3 and \cite{BobkovL00}, Proposition 3.1]
\label{lem:sc_conc}
Let $X \sim \pi$ for density $\pi: \xset \to \R$ which is $\mu$-strongly logconcave in $\normx{\cdot}$, and let $\ell: \xset \to \R$ be $G$-Lipschitz in $\normx{\cdot}$. For all $t \ge 0$,
\[\Pr_{x \sim \pi}\Brack{\ell(x) \ge \E_\pi[\ell] + t} \le \exp\Par{-\frac{\mu t^2}{2G^2}}.\]
\end{lemma}

In the remainder of the section, let $\tpi_y$ be the distribution of the output of Algorithm~\ref{alg:inner} and recall the target stationary distribution is $\pi_y$.
When $\rho$ is clear from context, we define $\brho \defeq \med(0, \rho, 2)$ to be the truncation of $\rho$ to $[0, 2]$. We also denote the index set drawn on Line~\ref{line:jiadraw} by
\[\jset \defeq \Brace{j_{i, b}}_{b \in [a], i \in [b]},\]
when $a$ is clear from context. We first provide the following characterization of $\tvd{\pi_y}{\tpi_y}$.

\begin{lemma}\label{lem:tvdrho}
Define $r_x$ to be the random variable $\E[\rho \mid x_1 = x]$ (where the expectation is over $x_2$, $a$, and the random indices $\jset$, and similarly let $\br_x \defeq \E[\brho \mid x_1 = x]$. Then,
\begin{align*}
    \tvd{\pi_y}{\tpi_y} \le \E_{x \sim \gamma_y}\Abs{r_x - \br_x}.
\end{align*}
\end{lemma}
\begin{proof}
First, by definition of $\pi_y$, we have
\begin{equation}\label{eq:piy_gamma}
\pi_y(x) = \frac{\exp(-F(x)) \gamma_y(x)}{\int \exp(-F(w)) \gamma_y(w) \dd w} = \gamma_y(x) \cdot \frac{\exp(-F(x))}{\E_{w\sim\gamma_y} \exp(-F(w))}.
\end{equation}
Moreover, by definition of the algorithm,
\begin{equation}\label{eq:tpiy_gamma}
\tpi_y(x) = \frac{\gamma_y(x) \Pr[u \le \half \rho\mid x_1 = x]}{\Pr[u \le \half \rho]} = \frac{\gamma_y(x) \E[\brho \mid x_1 = x]}{\E[\brho]}
\end{equation}
where all probabilities and expectations are $x_2$, $a$, and $\jset$. Furthermore, note that for fixed $b \in [a]$,
\[\E_{\jset}\Brack{\prod_{i \in [b]} (f_{\jia}(x_2) - f_{\jia}(x_1))} = \Par{\E_{j \sim \ind}\Brack{f_{j}(x_2) - f_{j}(x_1)}}^b = (F(x_2) - F(x_1))^b.\]
Hence, taking expectations over $a$, we have for any fixed $x_1$, $x_2$,
\begin{equation}\label{eq:rhoexpect}
\begin{aligned}
	\E\Brack{\rho \mid x_1, x_2} &= \sum_{b \ge 0} \Pr[a \ge b] (F(x_2) - F(x_1))^b \\
&=\sum_{b \ge0}\frac{1}{b!}(F(x_2) - F(x_1))^b = \exp\Par{F(x_2) - F(x_1)}.
\end{aligned}
\end{equation}
Next, by combining \eqref{eq:piy_gamma} and \eqref{eq:tpiy_gamma}, we have
\begin{align*}
\tvd{\pi}{\tpi} &= \half \int \Abs{\frac{\exp(-F(x))}{\E_{w \sim \gamma_y}\exp(-F(w))} - \frac{\E[\brho \mid x_1 = x]}{\E[\brho]}}\gamma_y(x)\dd x \\
&= \half \E_{x \sim \gamma_y}\Brack{ \Abs{\frac{\exp(-F(x))}{\E_{w \sim \gamma_y}\exp(-F(w))} - \frac{\E[\brho \mid x_1 = x]}{\E[\brho]}}}.
\end{align*}
By taking expectations over $x_2$ in \eqref{eq:rhoexpect},  and recalling the definitions of $r_x, \br_x$, we obtain $r_x=\E[\rho\mid x_1=x]=\exp(-F(x))\E_{x_2\sim\gamma_y}\exp(F(x_2))$.
We thus have
\[\tvd{\pi}{\tpi} = \half \E_{x \sim \gamma_y}\Brack{\Abs{\frac{r_x}{\E_{w \sim \gamma_y} r_w} - \frac{\br_x}{\E_{w \sim \gamma_y} \br_w}}}.\]
Next, we lower bound $\E_{w \sim \gamma_y} r_w$ as follows. By taking expectations over \eqref{eq:rhoexpect} and using independence of $x_1$ and $x_2$, we have that for the random variable $Z = \exp(-F(x))$ where $x \sim \gamma_y$, we have
\begin{equation}\label{eq:expectinvprod}
\E_{w \sim \gamma_y} r_w = \Par{\E Z} \cdot \Par{\E Z^{-1}} \ge 1,
\end{equation}
where we used Jensen's inequality which implies the last inequality for any nonnegative random variable $Z$. Finally, combining the above two displays, we derive the desired bound as follows:
\begin{align*}\half \E_{x \sim \gamma_y}\Brack{\Abs{\frac{r_x}{\E_{w \sim \gamma_y} r_w} - \frac{\br_x}{\E_{w \sim \gamma_y} \br_w}}} &\le \half\E_{x \sim \gamma_y}\Brack{\Abs{\frac{r_x}{\E_{w \sim \gamma_y} r_w} - \frac{\br_x}{\E_{w \sim \gamma_y} r_w}}} \\
&+ \half \E_{x \sim \gamma_y}\Brack{\Abs{\frac{\br_x}{\E_{w \sim \gamma_y} r_w} - \frac{\br_x}{\E_{w \sim \gamma_y} \br_w}}} \\
&\le \half \E_{x \sim \gamma_y}\Brack{\Abs{r_x - \br_x}} + \frac {\E_{x \sim \gamma_y}[\Abs{\br_x}]} 2 \cdot \Abs{\frac{1}{\E_{w \sim \gamma_y} \br_w} - \frac{1}{\E_{w \sim \gamma_y} r_w}} \\
&= \half \E_{x \sim \gamma_y}\Brack{\Abs{r_x - \br_x}} + \half \Abs{1 - \frac{\E_{x \sim \gamma_y} \br_x}{\E_{x \sim \gamma_y} r_x}} \\
&\le \half \E_{x \sim \gamma_y}\Brack{\Abs{r_x - \br_x}} + \frac 1 {2 |\E_{x \sim \gamma_y} r_x|}\cdot \E_{x \sim \gamma_y}\Brack{\Abs{r_x - \br_x}} \\
&\le \E_{x \sim \gamma_y}\Brack{\Abs{r_x - \br_x}}.
\end{align*}
In the second and last inequalities, we use the bound \eqref{eq:expectinvprod}. The third line follows since $\br_x$ is always nonnegative by definition, and the third inequality used convexity of $|\cdot|$.
\end{proof}

Lemma~\ref{lem:tvdrho} shows it remains to bound $\E_{x \sim \gamma_y} |r_x - \br_x|$. Fixing $x_1$ and $x_2$, we know $\rho$ and $\brho$ as random variables of $a$ and $\jset$ are equal, except for the effect of truncating $\rho$ to $[0, 2]$. Hence,
\begin{equation}\label{eq:rhobound}\E_{x \sim \gamma_y} |r_x - \br_x| \le \E [|\rho| \1_{\rho \not\in [0, 2]}].\end{equation}
In the remainder of the section, define
\begin{equation}\label{eq:Ldef}H \defeq \left\lceil10\log \frac 1 \delta \right\rceil.\end{equation}
We then let
\begin{equation}\label{eq:lamsigdef}
\begin{aligned}
\lam &\defeq \sum_{b > H} \1_{a \ge b} \prod_{i \in [b]} (f_{\jia}(x_2) - f_{\jia}(x_1)), \\
\sigma &\defeq \sum_{b = 0}^H \1_{a \ge b} \prod_{i \in [b]} (f_{\jia}(x_2) - f_{\jia}(x_1)),
\end{aligned}
\end{equation}
be random variables depending on the choices of $x_1, x_2, a, \jset$, where $\lam$ captures the effect of the ``large'' $b$, and $\sigma$ captures the effect of the ``small'' $b$ (where the $b = 0$ term is $1$ by convention). Since $\rho = \sigma + \lam$, in light of \eqref{eq:rhobound} it suffices to bound $\E[|\sigma|\1_{\rho \not\in [0, 2]}]+\E[|\lam|\1_{\rho \not\in [0, 2]}]$, as
\begin{align}
\label{eq:bound_sigma+lam}
    \E_{x \sim \gamma_y} |r_x - \br_x| \le \E [|\rho| \1_{\rho \not\in [0, 2]}]\le \E[|\sigma|\1_{\rho \not\in [0, 2]}]+\E[|\lam|\1_{\rho \not\in [0, 2]}].
\end{align}

We defer proofs of the following to Appendix~\ref{app:rejecthelper}, using small modifications to \cite{GLL22}. 

\begin{restatable}{lemma}{restatelargebound}\label{lem:largebound}
For $\lam$ defined in \eqref{eq:lamsigdef},
\[
\E\Brack{ |\lam| \1_{\rho \not\in [0, 2]}} \le \frac \delta 4.
\]
\end{restatable}

\begin{restatable}{lemma}{restatesmallbound}\label{lem:smallbound}
For $\sigma$ defined in \eqref{eq:lamsigdef},
\[\E\Brack{|\sigma| \1_{\rho \not\in [0, 2]}} \le \frac \delta 4.\]
\end{restatable}

Putting together these pieces, we finally obtain the following guarantee on Algorithm~\ref{alg:inner}.

\begin{proposition}
\label{prop:guarantee_of_inner}
The output of Algorithm~\ref{alg:inner} has total variation distance to $\pi_y$ bounded by $\delta$.
In expectation, Algorithm~\ref{alg:inner} queries $O(1)$ random $f_i$ and draws $O(1)$ samples from $\gamma_y$.
\end{proposition}
\begin{proof}
The total variation distance bound comes from combining Lemma~\ref{lem:tvdrho},~\eqref{eq:bound_sigma+lam}, Lemma~\ref{lem:largebound}, and Lemma~\ref{lem:smallbound}. Further, the end probability of each ``while" loop is $\Pr[u \le \half \rho]=\E[\brho]=\E_{x\sim\gamma}\br_x\ge \E_{x\sim\gamma_y}r_x-\E_{x \sim \gamma_y}|\br_x-r_x|$.
We proved in \eqref{eq:expectinvprod} that $\E_{x\sim\gamma_y}r_x\ge 1$, and combining \eqref{eq:bound_sigma+lam}, Lemma~\ref{lem:largebound} and Lemma~\ref{lem:smallbound}, shows $\E_{x \sim \gamma_y}|\br_x-r_x|\le \delta\le \half$.
Hence the expected number of loops is $\le 2$, and each loop draws two samples from $\gamma_y$, and $O(1)$ many $f_i$ in expectation since $\E a^2 = O(1)$.
\end{proof}

\subsection{Analysis of Algorithm~\ref{alg:alternatesample}}\label{ssec:mixing}

We now prove a mixing time on Algorithm~\ref{alg:alternatesample} using a standard conductance argument, by using tools developed in Section~\ref{sec:llt}. We first define our notion of conductance.

\begin{definition}\label{def:conductance}
	For a reversible Markov chain with stationary distribution $\pi$ supported on $\xset$ and transition distributions $\{\T_x\}_{x \in \xset}$, we define the conductance of the Markov chain by
	\begin{equation*}
		\Phi := \inf_{S \subset \xset} \frac{\int_S \T_x (\xset \backslash S) \dd \pi(x)}{\min \{\pi (S), \pi (\xset \backslash S))\}}.
	\end{equation*}
\end{definition}

We further recall a standard way of lower bounding conductance via isoperimetry.

\begin{lemma} [\cite{lee2018convergence}, Lemma 13]
	\label{lem:conductance_lower}
	In the setting of Definition~\ref{def:conductance}, let $d: \xset \times \xset$ be a metric on $\xset$. Suppose for any $x, x' \in \xset$ with $d(x,x')\leq \Delta$, 
	$$\tvd{\T_x}{\T_{x'}}\leq \frac 1 2. $$
	Also, suppose that for any partition $S_1, S_2, S_3$ of $\R^d$, $\pi$ satisfies the isoperimetric inequality 
	\begin{equation*}
		\pi(S_3)   \geq C_{\textup{iso}} \Par{\min_{x\in S_1, y\in S_2}d(x,y)} \min \Brace{\pi(S_1), \pi(S_2)}.
	\end{equation*}
	Then $\Phi = \Omega \Par{ \Delta C_{\textup{iso}} }$.
\end{lemma}

Finally, a classical result of \cite{lovasz1993random} shows how to upper bound mixing time via conductance. 
\begin{lemma}[\cite{lovasz1993random}, Corollary 1.5]
	\label{lem:tv_mixing}
	In the setting of Definition~\ref{def:conductance}, let $\pi_t$ be the distribution after $t$ steps of the Markov chain. If the starting distribution $\pi_0$ is $\beta$-warm with respect to $\pi$
	$$\tvd{\pi_t}{\pi} \leq \sqrt{\beta}\Par{1- \frac {\Phi^2} 2}^t.$$
\end{lemma}

Leveraging Lemmas~\ref{lem:conductance_lower} and~\ref{lem:tv_mixing}, we prove the following mixing time bound.

\begin{proposition}
	\label{prop:mixing_alternate}
Assume the input $x_0$ to Algorithm~\ref{alg:alternatesample} is drawn from a $\beta$-warm distribution with respect to $\pi$, $\eta \mu \le 1$, and $T = \Omega(\frac 1 {\eta^2\mu^2} \log \frac \beta \delta)$ for a sufficiently large constant. Then the output of Algorithm~\ref{alg:alternatesample} has total variation distance to $\pi$ bounded by $\delta$.
\end{proposition}
\begin{proof}
	Following the optimal coupling characterization of total variation, whenever the optimal coupling of $y \sim \Npsi_x$ and $y' \sim \Npsi_{x'}$ sets $y = y'$ in Line~\ref{line:ysample} of Algorithm~\ref{alg:alternatesample}, we can couple the resulting distributions in Line~\ref{line:xsample} as well. This shows that  $\tvd{\T_x }{\T_{x'}} \leq \|\Npsi_x - \Npsi_{x'}\|_{\textup{TV}} $. By Lemma~\ref{lem:psi_sc}, since $\vhi$ is convex, $\psi$ is a self-concordant function. Then, combined with Lemma~\ref{lem:tvdist}, for any $d_{\psi}(x,x')\leq \frac 1 4$, 
	\begin{equation*}
		\tvd{\T_x }{\T_{x'}} \leq \tvd{\Npsi_x}{\Npsi_{x'}} \leq  \frac 1 2.
	\end{equation*}
	By Lemma \ref{lem:iso},  since  $F+\eta \mu \psi$ is $ {\eta \mu}$-relatively strongly convex in $\psi$, $\pi$ satisfies the isoperimetric inequality such that for any partition $S_1, S_2, S_3$ of $\R^d$,
	\begin{align*}
		\pi(S_3) = \Omega(\eta\mu) \Par{\min_{x\in S_1, y\in S_2}d_{\psi}(x,y)} \min \Brace{\pi(S_1), \pi(S_2)}.
	\end{align*}
	By Lemma~\ref{lem:conductance_lower}, we can then lower bound the conductance by $\Phi = \Omega (\eta\mu)$. Choosing a sufficiently large constant in $T$, we conclude by Lemma~\ref{lem:tv_mixing} the desired
	$	\tvd{\pi_T}{\pi} \leq \sqrt{\beta} \exp (-\frac{T\Phi^2}{2}) \leq \delta.$
\end{proof}

By combining Proposition~\ref{prop:guarantee_of_inner} with Proposition~\ref{prop:mixing_alternate}, we can now complete our analysis.

\begin{restatable}{theorem}{restatesample}\label{thm:mainllt}
In the setting of Problem~\ref{prob:mainsetting}, let $\eta \mu \le 1$ and assume $x_0$ has a $\beta$-warm distribution with respect to $\pi$ defined in \eqref{eq:maindensity}. Further for sufficiently large constants suppose $\frac 1 \eta = \Omega(G^2\log \frac {\log\beta}{\delta\eta\mu})$ and
\[T = \Theta\Par{\frac 1 {\eta^2\mu^2} \log \frac {\beta}{\delta}}.\]
Algorithm~\ref{alg:alternatesample} using Algorithm~\ref{alg:inner} with error parameter $\frac \delta{2T}$ to implement Line~\ref{line:xsample} returns a point with $\delta$ total variation distance to $\pi$, querying $O(T)$ random $f_i$ in expectation.
\end{restatable}
\begin{proof}
Proposition~\ref{prop:mixing_alternate} guarantees that if each call to Line~\ref{line:xsample} of Algorithm~\ref{alg:alternatesample} is implemented exactly, we obtain $\frac \delta 2$ total variation to $\pi$. Further, the total variation error accumulated over $T$ calls to Algorithm~\ref{alg:inner} is less than $\frac \delta 2$ by a union bound on Proposition~\ref{prop:guarantee_of_inner}. Combining these bounds results in the desired total variation guarantee, and the complexity bound follows from Proposition~\ref{prop:guarantee_of_inner}.
\end{proof}

Theorem~\ref{thm:mainllt} is qualitatively analogous to previous mixing time results on the proximal sampler of \cite{lee2021structured}, e.g., Theorem 1, \cite{lee2021structured} or Theorem 4,  \cite{ChenCSW22}, with two major differences. 

First, Theorem~\ref{thm:mainllt} has a quadratically-worse dependence on the ``relative strong logconcavity parameter'' $\eta\mu$ than prior Euclidean specializations, giving a mixing time scaling as $\approx \frac 1 {\eta^2\mu^2}$ as opposed to the $\approx \frac 1 {\eta\mu}$ mixing times known in the Euclidean case. Second, because it uses the conductance machinery of Lemma~\ref{lem:tv_mixing} (which implicitly bounds $\chi^2$ decay to the stationary distribution), it scales logarithmically in the initial warmness parameter $\beta$. Stronger dependences (scaling as $\log\log\beta$) are known for the Euclidean setting, via bounding the log-Sobolev constant and associated relative entropy decay (see e.g., Theorem 1, \cite{lee2021structured} or Theorem 3, \cite{ChenCSW22} or Theorem 58, \cite{ChenE22}). We discuss both of these points in further depth as natural open problems in Section~\ref{sec:future}.

Finally, we note that given sample access to $\exp(-\eta\mu\psi(x))\1_{x \in \xset}$, a distribution which only depends on the choice of $\vhi$ and $\xset$ (and not the function $F$), we obtain $\beta \le \exp(GD)$ in Theorem~\ref{thm:mainllt}.

\begin{lemma}\label{lem:feasiblestart}
In the setting of Problem~\ref{prob:mainsetting}, the density $\nu$ satisfying 
\[\dd\nu(x) \propto \exp(-\eta\mu\psi(x))\1_{\xset}(x) \dd x\]
is $\exp(GD)$-warm for $\pi$ defined in \eqref{eq:maindensity}.
\end{lemma}
\begin{proof}
Note that for all $x, w \in \xset$, $|F(x) - F(w)| \le GD$. Further recall $\pi \propto \exp(-F) \nu$. We conclude by observing that for all $x \in \xset$,
\begin{align*}
\frac{\exp\Par{-F(x)} \nu(x)}{\int_{\xset} \exp(-F(w))\nu(w)\dd w} \cdot \frac{\int_{\xset} \nu(w)\dd w}{\nu(x)} = \frac{\int_{\xset} \nu(w)\dd w}{\int_{\xset} \exp(F(x)-F(w))\nu(w)\dd w} \le \exp\Par{GD}.
\end{align*}
\end{proof} 	%

\section{Applications}
\label{sec:apps}

In this section, we discuss applications of the sampling scheme we develop in Section~\ref{sec:sampling}. We begin by specializing our machinery to $\ell_p$ and Schatten-$p$ norms in Section~\ref{ssec:lp}. We then give new algorithms with improved zeroth-order query complexity for private convex optimization in Section~\ref{ssec:privacy}. Finally, in Section~\ref{ssec:valueoracle} we discuss computational issues regarding the specific LLT we introduce.

\subsection{LLT for $\ell_p$ and Schatten-$p$ norms}\label{ssec:lp}

Throughout this section we fix some $p \in [1, 2]$, and define the dual value $q \ge 2$ such that $\frac 1 q + \frac 1 p = 1$. It is well-known that the $\ell_q$ norm and $\ell_p$ norm are dual, as are the corresponding Schatten norms. In light of Lemma~\ref{lem:scllt}, to obtain a sampler catering to the $\ell_p$ geometry for example, it suffices to take the LLT of a smooth function in $\ell_q$. We provide the latter by recalling the following fact.

\begin{fact}\label{fact:lqsquare}
Let $p \in [1, 2]$, $q \ge 2$ satisfy $\frac 1 p + \frac 1 q$. If $\norm{\cdot}_q$ is a vector $\ell_q$ norm, $\frac 1 2 \norm{\cdot}_q^2$ is $\frac 1 {p - 1}$-smooth in the $\ell_q$ norm, and if $\norm{\cdot}_q$ is a matrix Schatten-$q$ norm, $\frac 1 2 \norm{\cdot}_q^2$ is $\frac 1 {p - 1}$-smooth in the Schatten-$q$ norm.
\end{fact}
\begin{proof}
This follows (for example) from three well-known facts: 1) that $\half \norm{\cdot}_q^2$ and $\half \norm{\cdot}_p^2$ are conjugate functions in both the vector and matrix cases, 2) that the conjugate of a $m$-strongly convex function in a norm is $\frac 1 m$-smooth in the dual norm \cite{KakadeST09}, and 3) that $\half \norm{\cdot}_p^2$ is $(p - 1)$-strongly convex in $\norm{\cdot}_p$ in both the vector and matrix cases \cite{BallCL94}.
\end{proof}

\paragraph{$\ell_p$ norms.} Next, for any $a > 0$, when the context is clearly about vector spaces, we define
\begin{equation}\label{eq:llt_lp}\psi_{p, a}(x) \defeq \log\Par{\int \exp\Par{\inprod{x}{y} - a\norm{y}_q^2}\dd y}. \end{equation}
Note that as the LLT of a $\frac {2a}{p - 1}$-smooth function in  $\ell_q$, $\psi_{p, a}$ is $\Omega(\frac {p - 1} a)$-strongly convex in $\ell_p$ by Lemma~\ref{lem:scllt}. In applications we fix a value of $\eta > 0$, set $a = \Theta((p - 1)\eta)$, and use $\eta \psi_{p, a}$ as our strongly convex regularizer in $\ell_p$. We next provide a bound on the range of $\psi_{p, a}$.

\begin{lemma}\label{lem:rangellt}
Let $a > 0$ and let $d \in \N$ be at least a sufficiently large constant. The additive range of $\psi_{p, a}$ over $\{x \in \R^d \mid \norm{x}_p \le 1\}$ is 
\[O\Par{1 + \frac 1 a + \sqrt{\frac d a \log\Par{a+ \frac d a}}}.\]
In particular, for $a \le \frac 1 {d\log d}$, the additive range is $O(\frac 1 a)$.
\end{lemma}
\begin{proof}
Throughout the proof denote for simplicity $\psi \defeq \psi_{p, a}$ and let 
\[\Npsi_x(y) \propto \exp\Par{\inprod{x}{y} - a\norm{y}_q^2}\]
be the associated density. By the characterization of $\nabla \psi$ in Lemma~\ref{lem:cumulant} and the fact that the associated density $\Npsi_x$ is symmetric in $y$ for $x = 0$, we have $\nabla \psi(0) = 0$ and hence it suffices to bound $\psi(x) - \psi(0)$ for $\norm{x}_q \le 1$. We simplify this expression as
\begin{equation}\label{eq:psirange}
\begin{aligned}
\psi(x) - \psi(0) &= \log\Par{\int \exp\Par{\inprod{x}{y} - a \norm{y}_q^2}\dd y} - \log\Par{\int \exp\Par{- a \norm{y}_q^2}\dd y} \\
&= \log\Par{\int \exp\Par{\inprod{x}{y}} \frac{\exp\Par{ -a \norm{y}_q^2}}{\int \exp\Par{- a \norm{y}_q^2}\dd y} \dd y} = \log\Par{\E_{y \sim \Npsi_0}\Brack{\exp\Par{\inprod{x}{y}}} }.
\end{aligned}
\end{equation}
Next, let $\pi$ be the probability density on $\R_{\ge 0}$ such that
\begin{align*}
\dd \pi(r) \propto r^{d - 1} \exp\Par{-a r^2} \dd r.
\end{align*}
We note $\dd\pi(r)$ is the density of the scalar quantity $r = \norm{y}_q$ for $y \sim \Npsi_0$. This can be seen by taking a derivative of the volume of the $\ell_p$ ball of radius $r$, which scales as $r^d$, so the surface area of the ball scales as $r^{d - 1}$.
By H\"older's inequality, $\inprod{x}{y} \le \norm{y}_q$ for all $y$, since $\norm{x}_p \le 1$. We then continue \eqref{eq:psirange} and bound $\psi(x) - \psi(0) \le \log( \E_{r \sim \pi} \exp(r))$, and the conclusion follows from Lemma~\ref{lem:annoyingintegral}.
\end{proof}

\begin{lemma}\label{lem:annoyingintegral}
For any $a > 0$ and $d \in \N$ at least a sufficiently large constant,
\[\log\Par{\frac{\int_0^\infty \exp\Par{(d - 1)\log r + r - a r^2} \dd r}{\int_0^\infty \exp\Par{(d - 1)\log r - a r^2} \dd r}} \le 8 + \frac 8 a + \sqrt{\frac{8d}{a}\log\Par{a + \frac d a}}.\]
\end{lemma}
\begin{proof}
Throughout this proof let
\[Z \defeq \int_0^\infty \exp\Par{(d - 1)\log r - \alpha r^2} \dd r = \frac{\Gamma(\frac d 2)}{2a^{\frac d 2}},\; \tau \defeq 7 + \frac 8 a + \sqrt{\frac {8d} a \log\Par{ a + \frac{d}{a}}}.\]
Next we split the numerator of the left-hand side into two integrals:
\begin{align*}
I_1 &\defeq \int_0^\tau \exp\Par{(d - 1)\log r + r -a r^2} \dd r,\\
I_2 &\defeq \int_\tau^\infty \exp\Par{(d - 1)\log r + r - a r^2} \dd r.
\end{align*}
It is immediate that $I_1 \le \exp(\tau) Z$. Further, we recognize that for $r \ge \tau$, 
\begin{align*}
\max\Par{r, (d - 1)\log r} \le \frac {a r^2} 4.
\end{align*}
The first piece in the maximum is clear from $\tau \ge \frac 4 a$. The second follows since $\frac{r^2}{\log r}$ is an increasing function for $r \ge 7$, and either $\frac{4d}{a} \le 10$ in which case we use $\frac{7^2}{\log 7} \ge 10$, or we let $C \defeq \frac{4d}{a}$ and use
\[\frac{r^2}{\log r} \ge C \text{ for } r \ge \sqrt{2C \log \frac C 4},\; C \ge 10.\]
Hence we may bound
\begin{align*}
I_2 \le \int_\tau^\infty \exp\Par{-\frac{a r^2}{2}} = \sqrt{\frac {2\pi} a} \Pr_{t\sim\Nor(0,  a^{-1})}[t \ge \tau] \le \frac{2}{a\tau} \exp\Par{-\frac {a \tau^2} 2}.
\end{align*}
Above, we used Mill's inequality
\[\Pr_{t \sim \Nor(0, \sigma^2)}\Brack{t \ge \tau} \le \sqrt{\frac 2 \pi} \frac \sigma \tau \exp\Par{-\frac{\tau^2}{2\sigma^2}}.\]
Further for our $\tau$, our upper bound on $I_1$ is larger than our upper bound on $I_2$. To see this, 
\begin{align*}
\tau\Par{1 + \frac {a\tau} 2} + \frac d 3 \log d \ge \frac d 2 \log a &\implies \exp\Par{\tau\Par{1 + \frac {a\tau} 2}}\Gamma\Par{\frac d 2} \ge a^{\frac d 2} \\
&\implies \frac{\exp\Par{\tau}\Gamma(\frac d 2)}{2a^{\frac d 2}} \ge \frac 4 {a\tau}\exp\Par{-\frac{a\tau^2}{2}}.
\end{align*}
The first inequality is because $a\tau^2 \ge d\log a$.
The first implication then follows by exponentiating and using $\log \Gamma(\frac d 2) \ge \frac d 3 \log d$ for sufficiently large $d$, and the second implication follows by rearranging and using $a\tau \ge 4$. Finally the conclusion follows from
\[\log\Par{\frac{\int_0^\infty \exp\Par{(d - 1)\log r + r - a r^2} \dd r}{\int_0^\infty \exp\Par{(d - 1)\log r - a r^2} \dd r}} \le \log\Par{\frac{2\exp(\tau)Z}{Z}} \le \tau + 1.\]
\end{proof}

\paragraph{Schatten-$p$ norms.} When the context is clearly about matrix spaces, we analogously define
\[\psi_{p, a}(\mx) \defeq \log\Par{\int \exp\Par{\inprod{\mx}{\my} - a \norm{\my}_q^2} \dd y}.\]
The proof of Lemma~\ref{lem:rangellt} implies the following analogous range bound in this setting.

\begin{corollary}\label{cor:rangelltmat}
Let $a > 0$ and let $d_1, d_2 \in \N$ be at least sufficiently large constants.  The additive range of $\psi_{p, a}$ over $\{\mx \in \R^{d_1 \times d_2} \mid \norm{\mx}_p \le 1\}$ is
\[O\Par{1 + \frac 1 a + \sqrt{\frac {d_1d_2} a \log\Par{a + \frac{d_1d_2}{a}}}}.\]
In particular, for $a \le \frac{1}{d_1d_2\log(d_1d_2)}$, the additive range is $O(\frac 1 a)$.
\end{corollary}

\subsection{Zeroth-order private convex optimization}\label{ssec:privacy}

In this section, we consider a pair of closely-related problems in private convex optimization. Let $\sset$ be a domain, and let $n \in \N$. We say that a mechanism (randomized algorithm) $\mech: \sset^n \to \Omega$ satisfies $(\eps, \delta)$-differential privacy (DP) if for any event $S \subseteq \Omega$ where $\Omega$ is the output space, and any two datasets $\data, \data' \in \sset^n$ which differ in exactly one element,
\[\Pr[\mech(\data) \in S] \le \exp(\eps)\Pr[\mech(\data') \in S] + \delta.\]
We next define the private optimization problems we study.

\begin{problem}[DP-ERM and DP-SCO]\label{prob:dpco}
Let $n \in \N$, $\eps, \delta \in (0, 1)$, $D, G \ge 0$, and let $\xset \subset \R^d$ be compact and convex with diameter in a norm $\normx{\cdot}$ at most $D$. Let $\prob$ be a distribution over a set $\sset$ such that for any $s \in \sset$, there is a $f(\cdot; s): \xset \to \R$ which is convex and $G$-Lipschitz in $\normx{\cdot}$. Let $\data \defeq \{s_i\}_{i \in [n]}$ consist of $n$ independent draws from $\prob$, and let $f_i \defeq f(\cdot; s_i)$ for all $i \in [n]$. 

In the \emph{differentially private empirical risk minimization (DP-ERM)} problem, we receive $\data$ and wish to design a mechanism $\mech$ which satisfies $(\eps, \delta)$-DP and approximately minimizes
\[\Ferm(x) \defeq \frac 1 n \sum_{i \in [n]} f_i(x).\]
In the \emph{differentially private stochastic convex optimization (DP-SCO)} problem, we receive $\data$ and wish to design a mechanism $\mech$ which satisfies $(\eps, \delta)$-DP and approximately minimizes
\[\Fpop(x) \defeq \E_{s \sim \prob}\Brack{f(x; s)}.\]
\end{problem}

The following powerful general-purpose result was proven in \cite{GopiLLST23} reducing the DP-ERM and DP-SCO problems to logconcave sampling problems catered to the $\normx{\cdot}$ geometry. We slightly improve the parameter settings used by Theorem 4 of \cite{GopiLLST23} for DP-SCO by noting that a smaller value of $k$ also suffices (due to the larger error bound), as observed by \cite{GLL22}.

\begin{proposition}[Theorem 3, Theorem 4, \cite{GopiLLST23}, Theorem 6.9, \cite{GLL22}]\label{prop:gennormdp}
In the setting of Problem~\ref{prob:dpco}, let $k \ge 0$, and let $r: \xset \to \R$ be $1$-strongly convex with respect to $\normx{\cdot}$, with additive range at most $\Theta$. Let $\nu$ be the density on $\xset$ satisfying $\dd \nu(x) \propto \exp(-k(\Ferm(x) + \mu r(x))) \1_{\xset}(x) \dd x$. Then the algorithm which returns a sample from $\nu$ for
\[k =\frac{\sqrt d n\eps}{G\sqrt{2\Theta \log \frac 1 {2\delta}}},\;  \mu = \frac{2G^2k\log \frac 1 {2\delta}}{n^2\eps^2},\]
satisfies $(\eps, \delta)$-DP, and guarantees
\begin{align*}
\E_{x \sim \nu}\Brack{\Ferm(x)} - \min_{x \in \xset}\Ferm(x) \le O\Par{G\sqrt{\Theta} \cdot \frac{\sqrt{d \log \frac 1 \delta}}{n\eps}}.
\end{align*}
Further, the algorithm which returns a sample from $\nu$ for
\[k = \frac{1}{G\sqrt{\Theta}} \cdot \sqrt{\Par{\frac{d\log \frac 1 {2\delta}}{\eps^2 n^2} + \frac{1}{n}}} \cdot \min\Par{\frac{\eps^2 n^2}{\log \frac 1 {2\delta}}, nd},\; \mu = G^2 k \cdot \max\Par{\frac{\log \frac 1 {2\delta}}{n^2\eps^2}, \frac{1}{nd}}\]
satisfies $(\eps, \delta)$-DP, and guarantees
\[\E_{x \sim \nu}\Brack{\Fpop(x)} - \min_{x \in \xset}\Fpop(x) \le O\Par{G\sqrt{\Theta} \cdot \Par{\frac{\sqrt{d\log \frac 1 \delta}}{n\eps} + \frac 1 {\sqrt n}}}.\]
\end{proposition}

Armed with Proposition~\ref{prop:gennormdp} and the sampler in Theorem~\ref{thm:mainllt}, we give our main results on Problem~\ref{prob:dpco}.

\begin{assumption}\label{assume:warmstart}
Fix $p \in [1, 2]$ and $k, a, \eta, \mu > 0$. In the setting of Problem~\ref{prob:dpco}, assume there is an algorithm $\alg$ which returns a point drawn from a $\beta$-warm start to the density $\nu$ satisfying
\[\dd\nu(x) \propto \exp\Par{-k\Par{\Ferm(x) + \eta\mu \psi_{p, a}(x)}}\1_{\xset}(x) \dd x.\]
\end{assumption}

\begin{restatable}{theorem}{restatedp}\label{thm:lp}
Let $p \in [1, 2]$, $\eps, \delta \in (0, 1)$. In the setting of Problem~\ref{prob:dpco} where $\normx{\cdot}$ is the $\ell_p$ norm on $\R^d$, there is an $(\eps, \delta)$-differentially private algorithm $\mech_{\textup{erm}}$ which produces $x \in \xset$ such that
\begin{align*}
\E_{\mech_{\textup{erm}}}\Brack{\Ferm(x)} - \min_{x \in \xset} \Ferm(x) &= O\Par{\frac{GD}{\sqrt{p - 1}} \cdot \frac{\sqrt{d \log \frac 1 \delta}}{n\eps}} \text{ for } p \in (1, 2], \\
\E_{\mech_{\textup{erm}}}\Brack{\Ferm(x)} - \min_{x \in \xset} \Ferm(x) &= O\Par{GD\sqrt{\log d} \cdot \frac{\sqrt{d \log \frac 1 \delta}}{n\eps}} \text{ for } p = 1.
\end{align*}
Further, there is an $(\eps, \delta)$-differentially private algorithm $\mech_{\textup{sco}}$ which produces $x \in \xset$ such that
\begin{align*}
	\E_{\mech_{\textup{sco}}}\Brack{\Fpop(x)} - \min_{x \in \xset} \Fpop(x) &= O\Par{\frac{GD}{\sqrt{p - 1}} \cdot \Par{\frac 1 {\sqrt n} + \frac{\sqrt{d \log \frac 1 \delta}}{n\eps}}} \text{ for } p \in (1, 2], \\
	\E_{\mech_{\textup{sco}}}\Brack{\Fpop(x)} - \min_{x \in \xset} \Fpop(x) &= O\Par{GD\sqrt{\log d} \cdot \Par{\frac 1 {\sqrt n} + \frac{\sqrt{d \log \frac 1 \delta}}{n\eps}}} \text{ for } p = 1.
\end{align*}
Both $\mech_{\textup{erm}}$ and $\mech_{\textup{sco}}$ call $\alg$ in Assumption~\ref{assume:warmstart}, appropriately parameterized, once. $\mech_{\textup{erm}}$ uses
\[O\Par{\Par{1 +\frac{n^2\eps^2}{\log \frac 1 \delta}}^2 \log^2\Par{\frac{(1 + n\eps)\log \beta}{\delta}} \log\frac \beta \delta}.\]
additional value queries to some $f(\cdot; s_i)$, and $\mech_{\textup{sco}}$ uses
\[O\Par{\min\Par{nd,\;1 +\frac{n^2\eps^2}{\log \frac 1 \delta}}^2 \log^2\Par{\frac{(1 + n\eps)\log \beta}{\delta}} \log\frac \beta \delta}\]
additional value queries to some $f(\cdot; s_i)$.
\end{restatable}
\begin{proof}
First, we slightly simplify the setting of Problem~\ref{prob:dpco}. We may first assume that $D = 1$, i.e.\ $\xset$ has diameter at most $1$ in $\normx{\cdot}$. If the diameter is bounded by some $D \neq 1$, we can rescale the domain $\xset \gets \frac 1 D \xset$, and remap to the modified functions $f(x; s) \gets f(Dx; s)$ over this modified domain for all $s \in \sset$. It is clear the Lipschitz constant rescales as $G \gets GD$ as a result. Next, we assume $(n\eps)^2 \ge d\Theta \log \frac 1 \delta$ where $\Theta = \min(\frac 1 {p - 1}, \log d)$. In the other case, in light of the diameter bound on $\xset$ and the Lipschitz assumption, returning a random point in $\xset$ attains the error bound claimed. Finally, assume $p \in (1, 2]$, as otherwise we set $p \gets 1 + \frac 1 {\log d}$, which only affects bounds by constant factors, since $\norm{\cdot}_p$ is affected by $O(1)$ multplicatively everywhere under this change.

Under these simplifications, we choose the parameters $k$ and $\mu$ according to Proposition~\ref{prob:dpco} for each problem. Assume for now that $\Theta$ for the regularizer $r$ we choose is bounded by a universal constant times $\frac 1 {p - 1}$. Then the Lipschitz constant of $k\Ferm$ in either case of Proposition~\ref{prob:dpco} is 
\[kG = \Omega\Par{\min\Par{\frac{\sqrt{(p - 1)d} n\eps}{\sqrt{\log \frac 1 \delta}}, d\sqrt{n} } } = \Omega(d),\]
as implied by our earlier simplification. We hence may choose $\ind$ to be uniform over $[n]$, and
\[\eta = O\Par{\frac{1}{k^2G^2 \log \frac {(1 + n\eps)\log \beta} \delta}}\]
for a sufficiently small constant to use Theorem~\ref{thm:mainllt}. Under this setting we certainly have $\eta = O(\frac 1 {d^2})$, so letting $r \defeq \eta \psi_{p, a}$ for $a \defeq \frac{\eta(p - 1)}{2}$ shows that $r$ is $\eta$ times the LLT of an $\eta$-smooth function in $\ell_q$. By Lemma~\ref{lem:scllt}, $r$ is indeed $1$-strongly convex in $\ell_p$, and Lemma~\ref{lem:rangellt} bounds its range by $\Theta = O(\frac 1 {p -1})$ satisfying our earlier assumption, where we use $a = O(\frac 1 {d^2})$. The runtime finally follows by applying our choices of $k, \mu$ in Proposition~\ref{prop:gennormdp}, with our choice of $\eta$, in Theorem~\ref{thm:mainllt}, where we ensure that $\eta \cdot k\mu \le 1$ by choosing a smaller $\eta$ if this is not the case (so Theorem~\ref{thm:mainllt} applies). Finally, to account for total variation error in our sampler, it suffices to adjust the failure probability $\delta$ by a constant and take a union bound over the privacy definition and the failure of Theorem~\ref{thm:mainllt}.
\end{proof}

We briefly compare Theorem~\ref{thm:lp} to its counterpart result, Corollary 2, in \cite{GopiLLST23}. For simplicity, in this discussion we let $\eps = \Theta(1)$ and $p = 1$ (similar qualitative comparisions hold for any $p \in [1, 2)$). 
In general, the two resulting value oracle query complexities are incomparable: compared to the information-theoretic lower bound derived in Appendix~\ref{app:lb}, Theorem~\ref{thm:lp} has a near-quadratic overhead (and also requires a warm start), whereas Corollary 2 in \cite{GopiLLST23} suffers from an $\approx d$ factor overhead. For small sample sizes $n \lesssim \sqrt{d}$, Theorem~\ref{thm:lp} is preferable, whereas for large sample sizes or when a warm start is not available, Corollary 2 in \cite{GopiLLST23} is preferable. 

However, \cite{GopiLLST23} relies on norm comparisons between $\ell_p$ and $\ell_2$ norms (as it calls the Euclidean proximal sampler), and thus a $\textup{poly}(d)$-factor overhead appears inherent to their approach. On the other hand, if Theorem~\ref{thm:mainllt} can be qualitatively improved as discussed in Section~\ref{sec:future}, it could close the query complexity gap to the information-theoretic lower bound. We believe this distinction highlights the novelty of our framework, and opens the door to interesting future improvements.

Finally, by combining the proof strategy of Theorem~\ref{thm:lp} with Corollary~\ref{cor:rangelltmat} instead of Lemma~\ref{lem:rangellt}, we immediately obtain the following corollary in the case of Schatten norms.

\begin{corollary}\label{cor:schatten}
	Let $p \in [1, 2]$, $\eps, \delta \in (0, 1)$. In the setting of Problem~\ref{prob:dpco} where $\normx{\cdot}$ is the Schatten-$p$ norm on $\R^{d_1 \times d_2}$, there is an $(\eps, \delta)$-differentially private algorithm $\mech_{\textup{erm}}$ which produces $\mx \in \xset$ such that
	\begin{align*}
		\E_{\mech_{\textup{erm}}}\Brack{\Ferm(\mx)} - \min_{\mx \in \xset} \Ferm(\mx) &= O\Par{\frac{GD}{\sqrt{p - 1}} \cdot \frac{\sqrt{d_1d_2 \log \frac 1 \delta}}{n\eps}} \text{ for } p \in (1, 2], \\
		\E_{\mech_{\textup{erm}}}\Brack{\Ferm(\mx)} - \min_{\mx \in \xset} \Ferm(\mx) &= O\Par{GD\sqrt{\log (d_1d_2)} \cdot \frac{\sqrt{d_1d_2 \log \frac 1 \delta}}{n\eps}} \text{ for } p = 1.
	\end{align*}
	Further, there is an $(\eps, \delta)$-differentially private algorithm $\mech_{\textup{sco}}$ which produces $\mx \in \xset$ such that
	\begin{align*}
		\E_{\mech_{\textup{sco}}}\Brack{\Fpop(\mx)} - \min_{\mx \in \xset} \Fpop(\mx) &= O\Par{\frac{GD}{\sqrt{p - 1}} \cdot \Par{\frac 1 {\sqrt n} + \frac{\sqrt{d_1d_2 \log \frac 1 \delta}}{n\eps}}} \text{ for } p \in (1, 2], \\
		\E_{\mech_{\textup{sco}}}\Brack{\Fpop(\mx)} - \min_{\mx \in \xset} \Fpop(\mx) &= O\Par{GD\sqrt{\log (d_1d_2)} \cdot \Par{\frac 1 {\sqrt n} + \frac{\sqrt{d_1d_2 \log \frac 1 \delta}}{n\eps}}} \text{ for } p = 1.
	\end{align*}
	Both $\mech_{\textup{erm}}$ and $\mech_{\textup{sco}}$ call $\alg$ in Assumption~\ref{assume:warmstart}, appropriately parameterized, once. $\mech_{\textup{erm}}$ uses
	\[O\Par{\Par{1 +\frac{n^2\eps^2}{\log \frac 1 \delta}}^2 \log^2\Par{\frac{(1 + n\eps)\log \beta}{\delta}} \log\frac \beta \delta}.\]
	additional value queries to some $f(\cdot; s_i)$, and $\mech_{\textup{sco}}$ uses
	\[O\Par{\min\Par{nd_1d_2,\;1 +\frac{n^2\eps^2}{\log \frac 1 \delta}}^2 \log^2\Par{\frac{(1 + n\eps)\log \beta}{\delta}} \log\frac \beta \delta}\]
	additional value queries to some $f(\cdot; s_i)$.
\end{corollary}

\subsection{Oracle access for $\psi_{p, a}$}\label{ssec:valueoracle}

In Theorem~\ref{thm:lp} and Corollary~\ref{cor:schatten}, we only bounded the value oracle complexity of our sampling algorithms. The remainder of the steps in Algorithm~\ref{alg:alternatesample} and its subroutine Algorithm~\ref{alg:inner} require samples from densities of the form $\dd\pi_x$ (for some $x \in \xset$) or $\dd\gamma_y$ (for some $y \in \R^d$), defined in \eqref{eq:marginaldef} and \eqref{eq:gammadef} respectively and reproduced here for convenience:
\begin{equation}\label{eq:reproduceoracles}
\begin{aligned}
\dd \pi_x(y) &= \exp\Par{\inprod{x}{y} - \psi(x) - \vhi(y)} \dd y, \\
\dd \gamma_y(x) &\propto \exp\Par{-\eta\mu\psi(x) - \Par{\psi(x) - \inprod x y}} \1_{\xset}(x) \dd x.
\end{aligned}
\end{equation}
These densities are independent of the function $F$ in Problem~\ref{prob:mainsetting} and hence do not require additional value oracle queries in the setting of Problem~\ref{prob:mainsetting}. In general, the complexity of these steps depends on the complexity of the functions $\vhi$ and $\psi$, and the set $\xset$. We now discuss strategies for sampling from $\pi_x$ and $\gamma_y$ in specific settings described by Section~\ref{ssec:lp}, which we first briefly summarize.

\begin{enumerate}
	\item We describe a method based on the inverse Laplace transform for sampling from $\pi_x$ and evaluating $\psi_{p, a}$ with complexity linear in the dimension $d$ in the vector setting.
	\item Under efficient value oracle access to $\psi_{p, a}$ and membership oracle access to $\xset$, general-purpose results \cite{LovaszV07, JiaLLV21, JambulapatiLV22} imply polynomial-time samplers for $\gamma_y$.
	\item We discuss generalizations of these methods to the matrix setting, and na\"ive sampling methods. We draw a loose connection to the HCIZ integral from harmonic analysis, and suggest how it may potentially help in the structured sampling task for LLTs in Schatten norms.
\end{enumerate}

\paragraph{$\ell_p$ setting.} We first discuss the case when $\xset \subset \R^d$ is a set on vectors equipped with the $\ell_p$ norm for some $p \in [1, 2]$, and we let $q \ge 2$ satisfy $\frac 1 p + \frac 1 q = 1$. We follow the notation \eqref{eq:llt_lp}. 

In order to sample from the density $\pi_x$, we use an \emph{inverse Laplace transform} decomposition. For a parameter $c \in [0, 1)$, we define the density $\mu_c$ supported on $\R_{\ge 0}$, such that for all $t \ge 0$,
\begin{equation}\label{eq:mulamdef}\exp(-t^c) = \int_0^\infty \exp\Par{-\lam t}\mu_c(\lam)\dd \lam.\end{equation}
Intuitively, the density $\mu_c(\lam)$ and the corresponding decomposition (inverse Laplace transform) \eqref{eq:mulamdef} aims to express the more heavy-tailed function $\exp(-t^c)$ as a distribution over the lighter-tailed functions $\exp(-\lam t)$. The inverse Laplace transform densities $\mu_c$ are well-studied in the probability theory literature, and correspond to \emph{stable count distributions} parameterized by $c$. For example, it is well-known that $\mu_{\half}$ is the L\'evy distribution
\[\dd\mu_{\half}(\lam) = \frac{1}{2\sqrt \pi \lam^{\frac 3 2}} \exp\Par{-\frac 1 {4\lam}}\dd \lam .\]
We refer the reader to references e.g.\ \cite{Mainardi07} on properties of the densities $\mu_c$, and for now assume we can access and sample from these one-dimensional distributions in closed form for simplicity. Given this decomposition, we can then write
\begin{equation}\label{eq:product1d}
\begin{aligned}
\exp(\psi_{p, a}(x)) &= \int \exp\Par{\inprod{x}{y} - a\norm{y}_q^2} \dd y \\
&= \int_0^\infty \Par{\int \exp\Par{\inprod{x}{y} - \lam a^{\frac q 2}\norm{y}_q^q}\dd y} \mu_{\frac 2 q}(\lam) \dd \lam \\
&= \int_0^\infty \prod_{i \in [d]} \Par{\int_{-\infty}^\infty \exp\Par{x_i y_i - \lam a^{\frac q 2} y_i^q} \dd y_i } \mu_{\frac 2 q}(\lam) \dd \lam .
\end{aligned}
\end{equation}
The decomposition \eqref{eq:product1d} reduces the problem of sampling from $\pi_x$ to $d$ one-dimensional problems. To sample $\propto \exp(\inprod{x}{y}-a\norm{y}_q^2)$, we can first sample $\lam$ from the density $\mu_c$ for $c = \frac 2 q$, and then sample each coordinate $y_i$ proportionally to $\exp(x_i y_i - \lam a^{\frac q 2}y_i^q)$ conditioned on the sampled $\lam$. 

This decomposition also gives us an efficient value oracle for $\psi_{p, a}$, by evaluating \eqref{eq:product1d} as a one-dimensional integral over $\lam$, where the integrand may be evaluated as a product of $d$ one-dimensional integrals. Under membership oracle access to $\xset$, the problem of sampling from $\gamma_y$ then falls under a generic logconcave sampling setup studied in a long line of work building upon \cite{DyerFK91}. The state-of-the-art general-purpose logconcave sampler, which combines the algorithms of \cite{LovaszV07, JiaLLV21} with the isoperimetric bound in \cite{JambulapatiLV22} (improving recent breakthroughs by \cite{Chen21, KlartagL22}), requires roughly $d^{3.5}$ value oracle calls to $\psi_{p, a}$ and membership oracle calls to $\xset$.

In principle, for structured sets $\xset$ (such as $\ell_p$ balls), the particular explicit structure of $\psi_{p, a}$ and $\xset$ may be exploited to design more efficient samplers for the densities $\gamma_y$, analogously to our custom linear-time sampler for $\pi_x$. However, it should be noted that the sampling problem for $\gamma_y$ appears to be quite a bit more challenging than the problem for $\pi_x$. We leave the investigation of explicit sampler design for $\gamma_y$ as an interesting open problem for future work.

\paragraph{Schatten-$p$ setting.} The situation is somewhat less straightforward in the matrix case. Here, the key computational problem in replicating the strategy suggested by \eqref{eq:product1d} is evaluating the integral
\begin{equation}\label{eq:badmatrixintegral}\int \exp\Par{\inprod{\mx}{\my} - C \norm{\my}_q^q} \dd \my,\end{equation}
where the integral is over $\my \in \R^{d_1 \times d_2}$, and $\mx \in \R^{d_1 \times d_2}$, $C > 0$ are fixed. The difficulty is $\inprod{\mx}{\my}$ decomposes coordinatewise, whereas $\norm{\my}_q^q$ decomposes spectrally.\footnote{Note that because $\norm{\cdot}_q$ is unitarially invariant, we may assume $\mx$ is diagonal.} At least superficially, this is similar to the challenge faced when evaluating the Harish-Chandra-Itzykson-Zuber (HCIZ) formula
\begin{equation}\label{eq:hciz}\int \exp\Par{\Tr\Par{\ma \mmu \mb \mmu^\dagger}}\dd \mmu,\end{equation}
where the integral is over the Haar measure on (complex) unitary matrices $\mmu$, and $\ma$, $\mb$ are Hermitian. By dropping the $-C\norm{\my}_q^q$ term in \eqref{eq:badmatrixintegral} and only integrating over unitary conjugations of a fixed matrix $\my$, we arrive at a generalization of \eqref{eq:hciz}. The difficulty in evaluating \eqref{eq:hciz} is also a sort of tension between the eigenspaces of $\ma$ and $\mb$. Nonetheless, \eqref{eq:hciz} has a (polynomial-time computable) exact formula, which was famously discovered independently by \cite{Harish-Chandra57, ItzyksonZ80}. Furthermore, \cite{LeakeMV21} recently obtained a polynomial-time sampler for the density induced by \eqref{eq:hciz}; while a sampler for \eqref{eq:badmatrixintegral} would follow from logconcavity and general-purpose results, it would be far from cheap, so ways of exploiting structure are fruitful to explore.

As a proof-of-concept, evaluating the integral \eqref{eq:badmatrixintegral} in (polynomial-time computable) closed form is a minimal requirement for implementing the $\mx$-oracles in \eqref{eq:reproduceoracles} used by our algorithm. Even this problem appears challenging, but (as summarized cleanly by \cite{Tao13, McSwiggen18}) a plethora of techniques exist for proving the HCIZ formula, some based on tools from stochastic processes. We pose the efficient computability of the integral \eqref{eq:badmatrixintegral} as another explicit open question. 	%

\section{Conclusion}
\label{sec:future}

We believe our work is a significant step towards developing the theory of LLTs and paving the way for their use in designing sampling algorithms. There are a number of important questions left open by our work, which we find interesting and potentially fruitful for the community to explore.

\paragraph{Stronger mixing time bounds.} Perhaps the most immediate open question regarding our alternating sampling framework in Section~\ref{sec:sampling} is to obtain a better understanding of its mixing time. 

The main qualitative shortcoming of our framework is that it yields a mixing time in Theorem~\ref{thm:mainllt} that has a inverse-quadratic dependence on the ``relative strong convexity'' parameter $\eta\mu$. This is in contrast with known mixing time results in the Euclidean setting, which more typically scale as $\approx \frac 1 {\eta\mu}$ \cite{lee2021structured, ChenCSW22, ChenE22}. We find it to be an interesting and important open problem to design a non-Euclidean framework for proximal sampling that closes this gap, giving a mixing time that scales linearly with $\frac 1 {\eta\mu}$, the natural analog to existing Euclidean counterparts.

Moreover, as discussed in Section~\ref{ssec:results}, Theorem~\ref{thm:mainllt}'s mixing time scales linearly in $\log \beta$, which as demonstrated by Lemma~\ref{lem:feasiblestart} (and related other settings, e.g.\ MALA \cite{ChewiLACGR21, lee2021structured}) can result in additional polynomial overhead in problem parameters: for what $\vhi, \psi$ is this avoidable? Notably, it is avoided for the Euclidean proximal sampler \cite{lee2021structured} by working directly with KL divergence (as opposed to the larger $\chi^2$ distance typically used by proofs using conductance bounds). Different proofs of this $\log\log\beta$ dependency for the Euclidean proximal sampler were then subsequently obtained by \cite{ChenCSW22, ChenE22}. We also mention that $\log\log\beta$ dependences may sometimes follow via average conductance techniques (e.g.\ \cite{LovaszK99}), which may apply to our Markov chain.

\paragraph{Samplers for explicit distributions.} Our results Theorem~\ref{thm:mainllt} and~\ref{thm:lp} mainly focused on bounding the query complexity to the function $F$, or samples $f_i$ from the distribution defining it. The total computational complexity of a practical implementation of Algorithm~\ref{alg:alternatesample} also includes the cost of sampling from the distributions \eqref{eq:reproduceoracles}, which are ``data-independent'' for this problem (only depending on explicit functions and sets instead of $F$). In Section~\ref{ssec:valueoracle}, we give a linear-time sampler for $\pi_x$ and a polynomial-time sampler for $\gamma_y$ under the $\ell_p$ geometry, but it is interesting to obtain faster samplers for particular structured choices of $(\vhi, \xset)$ of importance in applications.

\paragraph{LLT beyond proximal sampling.} More generally, we believe it is worthwhile to obtain a better understanding of specific choices of $(\vhi, \psi)$, e.g.\ the examples in Section~\ref{ssec:lp}, from an algorithmic perspective. LLTs satisfy appealing properties such as self-concordance, strong convexity, and isoperimetry making them well-suited for frameworks beyond Algorithm~\ref{alg:alternatesample}, such as discretized MLD \cite{AhnC21} and Metropolized sampling methods discussed in Section~\ref{sec:intro}. Bounding the complexity of their use in these applications necessitates an improved understanding of specific LLTs.

\paragraph{LLT as a dual object.} Finally, a tantalizing open question in the theory of well-conditioned sampling (even in the $\ell_2$ setting) is whether acceleration is achievable, i.e.\ mixing times scaling with the square root of the condition number (which is famously possible in optimization \cite{Nesterov83}). The duality of Fenchel conjugates appears to play a key role in acceleration, as made explicit by \cite{WangA18, CohenST21}, so a better understanding of duality may be helpful in the corresponding endeavor for sampling. The LLT is a natural candidate for a dual object in sampling, as it arises via joint densities on an extended space \eqref{eq:gen_joint}, and satisfies properties such as strong convexity-smoothness duality. Can we demystify this relationship, and use it to obtain faster samplers? 	
	\subsection*{Acknowledgments}
	We thank Sam Power for helpful comments, suggesting several references, noticing that an improved strong convexity-smoothness duality result follows from combining our smoothness-strong convexity duality result and a tool in \cite{chewipooladian2022caffarelli}, and letting us include his elegant observation. We also thank Vishwak Srinivasan and Yunbum Kook for bringing to our attention an error in Lemma~\ref{lem:iso}, which weakens the complexities in Theorems~\ref{thm:mainllt} and~\ref{thm:lp} by roughly a quadratic factor (cf.\ Section~\ref{ssec:erratum}).
	
	\bibliographystyle{alpha}	
\newcommand{\etalchar}[1]{$^{#1}$}

	\newpage
	\begin{appendix}
\section{Information-theoretic lower bound}\label{app:lb}

In this section, we show that prior information-theoretic lower bounds from \cite{DJWW15} and \cite{GLL22} can be straightforwardly extended to the settings studied by this paper to show that the value oracle complexities used by our algorithms in Section~\ref{sec:apps} are near-optimal, up to a quadratic factor. 

The main implication of these results is that if the gaps suggested by the first set of open problems in Section~\ref{sec:future} can be closed (i.e., linear rather than quadratic dependence on the inverse relative strong convexity, and mixing in relative entropy rather than relative variance), then we would unconditionally obtain near-optimal tradeoffs between zeroth-order query complexity and excess risk for private convex optimization in $\ell_p$ norms for $p \in (1, 2)$.

We first recall some notation from prior work and summarize previous results we will leverage.

\paragraph{Setup.} We consider the setting of stochastic optimization where there is a distribution over distributions $\{\prob_v\}_v$ indexed by $v$. An index $v$ is randomly selected, and we consider algorithms interacting with $\prob_v$ in one of two different ways. Letting $k \in \N$ and $\xset \subset \R^d$, \cite{DJWW15} defined a family of algorithms $\bbA_k$ such that $\alg \in \bbA_k$ can (adaptively) query a sequence of $k$ values $f(x; s)$ where $x \in \xset$ and $s$ is a fresh random sample from $\prob_v$. The follow-up work \cite{GLL22} defined another family of algorithms $\bbB_k$ which takes as input a dataset $\data = \{s_i\}_{i \in [n]}$ and can (adaptively) query a sequence of $k$ values $f(x; s)$ where $x \in \xset$ and $s \in \data$. These algorithm families model the SCO and ERM problems stated in Problem~\ref{prob:dpco}, without the privacy requirement. In a slight abuse of notation, we denote the output of an algorithm $\alg \in \bbA_k \cup \bbB_k$ in a SCO or ERM problem corresponding to a distribution $\prob$ by $\alg(\prob)$, where $\alg \in \bbB_k$ also depends on the dataset received.

Both \cite{DJWW15, GLL22} let $v$ be drawn uniformly at random from $\cV \defeq \{-1, 1\}^d$ and let
\[\prob_v \defeq \Nor\Par{\kappa v, \sigma^2 \id_d},\; f(x; s) \defeq \inprod{s}{x}\]
for parameters $\kappa, \sigma$ to be chosen. We fix this notation throughout this section. For any algorithm $\alg \in \bbA_k \cup \bbB_k$ corresponding to a set $\xset$ and a distribution $\prob$, we define the optimality gap
\[\eps_k(\alg, \xset, \prob) \defeq \E\Brack{\E_{s \sim \prob} f(\alg(\prob); s)} - \min_{x \in \xset} \E_{s \sim \prob} f(x; s), \]
where the first outer expectation is over any randomness in $\alg$, as well as in the samples used. We also define the minimax risk over a family of distributions $P$,
\[\eps^\star_k(\bbA_k \cup \bbB_k, P, \xset) \defeq \inf_{\alg \in \bbA_k \cup \bbB_k} \sup_{\prob \in P} \eps_k\Par{\alg, \prob, \xset}. \]
For $p \in [1, 2]$, we let $P_{G, p}$ denote the family of distributions $\prob$ over vectors $s$ such that \[\E_{s \sim \prob}\norm{s}_q^2 \le G^2,\text{ where } \frac 1 p + \frac 1 q = 1.\]
Our lower bounds in this section will be on $\eps^\star_k(\bbA_k \cup \bbB_k, P_{G, p}, \xset)$, where $\xset$ is a scaled $\ell_p$ ball. The family $P_{G, p}$ induces random linear functions $\inprod{s}{\cdot}$ with gradient $s$, and hence $\prob \in P_{G, p}$ implies that the induced function $\E_{s \sim \prob}\inprod{s}{\cdot}$ has a bounded-variance gradient oracle in the $\ell_p$ norm via queries to $\prob$. We use the following facts from prior work in our proofs.

\begin{lemma}[Section 5.1, \cite{DJWW15}]\label{lem:boundbyham}
Let $\xset$ be the $\ell_p$ ball of diameter $D$ for $p \in [1, 2]$. For any $v \in \cV$ and $x \in \xset$, letting $x^\star_v \defeq \min_{x \in \xset} \E_{s \sim \prob_v} f(x; s)$, and letting $\1(\sign(a)= \sign(b))$ be the $0$-$1$ function which is $1$ if and only if the signs of $a$ and $b$ agree,
\[\E_{s \sim \prob_v}\Brack{f(x; s)} - \E_{s \sim \prob_v}\Brack{f(x^\star_v; s)} \ge \frac{(1 - \frac 1 p)\kappa D}{2d^{\frac 1 p}} \sum_{j \in [d]} \1\Par{\sign(x_j) = \sign(v_j)}.\]
\end{lemma}

Lemma~\ref{lem:boundbyham} shows that it suffices to lower bound the expected Hamming distance between the signs of an estimate $x$ and a randomly sampled $-v$. Such a lower bound was given in \cite{DJWW15, GLL22} for estimates returned by $\alg \in \bbA_k \cup \bbB_k$ via information-theoretic arguments.

\begin{lemma}[Section 5.1, \cite{DJWW15}, Lemma 7.4, \cite{GLL22}]\label{lem:hamlb}
Let $\xset$ be the $\ell_p$ ball of diameter $D$, and let $\alg \in \bbA_k \cup \bbB_k$ be parameterized by $\xset$ and $\prob_v$. Then
\[\E_{v \sim_{\textup{unif.}} \cV}\Brack{\sum_{j \in [d]}\1(\sign(\alg(\prob_v)_j) = \sign(v_j)) } \ge \frac d 2 \Par{1 - \frac{\kappa \sqrt k}{\sigma \sqrt d} }.\]
\end{lemma}

To lower bound the oracle query complexity of our sampler we use the following standard result.

\begin{lemma}[\cite{DKL18}, Corollary 1]
	\label{lem:sampling_error_d}
	Let $\xset \subset \R^d$ be compact and convex, $f: \xset \to \R$ be convex, $k > 0$, and $\pi$ be the density over $\xset$ proportional to $\exp(-kf)$. Then,
	\begin{align*}
		\E_{x\sim\pi}[f(x)]-\min_{x\in\xset}f(x)\le \frac d k.
	\end{align*}
\end{lemma}

\paragraph{Lower bounds.} We now state three lower bounds generalizing results from \cite{DJWW15, GLL22}. Our results follow straightforwardly from Lemmas~\ref{lem:boundbyham},~\ref{lem:hamlb}, and~\ref{lem:sampling_error_d} with appropriate parameters. 

\begin{proposition}[Minimax risk lower bound, $P_{G, p}$]\label{prop:mmrisk}
Let $G, D > 0$, and let $p \in [1, 2]$, $q \ge 2$ satisfy $\frac 1 p + \frac 1 q = 1$. Let $\xset$ be the $\ell_p$ ball of diameter $D$. Then,
\[\eps_k^\star \Par{\bbA_k \cup \bbB_k, P_{G, p}, \xset} = \Omega\Par{GD\max\Par{1 - \frac 1 p, \frac{1}{\log d}}\min\Par{1, \sqrt{\frac d {k\log d}}}}.\]
\end{proposition}
\begin{proof}
Throughout the proof, let $\kappa = \frac{\sigma \sqrt d}{2 \sqrt k}$, and let
\begin{equation}\label{eq:sigmachoice}\sigma = \frac{G d^{- \frac 1 q}}{\sqrt{\frac d k + 4\log d}}.\end{equation}
By well-known bounds on the expected maximum of $d$ standard Gaussians, we have
\begin{align*}
	\E_{s \sim \prob_v}\Brack{\norm{s}_q^2} &\le 2\kappa^2 \norm{v}_q^2 + 2\E_{u \sim \Nor(0, \sigma^2 \id_d)}\Brack{\norm{u}_q^2} \\
	&\le 2\kappa^2 d^{\frac 2 q} + 2d^{\frac 2 q} \E_{u \sim \Nor(0, \sigma^2 \id_d)}\Brack{\norm{u}_\infty^2} \\
	&\le \sigma^2 d^{\frac 2 q} \Par{\frac d k + 4\log d} \le G^2.
\end{align*}
Hence, $\prob_v \in P_{G, p}$ for all $v \in \cV$, so it suffices to lower bound $\eps_k(\alg, \prob_v, \xset)$. Combining Lemmas~\ref{lem:boundbyham} and~\ref{lem:hamlb} with our choices of parameters,
\[\eps_k(\alg, \prob_v, \xset) \ge \frac{(1 - \frac 1 p)\kappa D d^{1 - \frac 1 p}}{8} = \Omega\Par{GD\Par{1 - \frac 1 p}\min\Par{1, \sqrt{\frac d {k\log d}}}}.\]
The conclusion then follows because for $p \le 1 + \frac 1 {\log d}$, choosing a larger value of $p$ only affects problem parameters by constant factors by norm conversions.
\end{proof}
We give a slight extension of Proposition~\ref{prop:mmrisk} for the family $\bP_{G, p}$ of distributions over linear functions $\inprod{s}{\cdot}$, where $s$ is required to satisfy $\norm{s}_q \le G$ with probability $1$, by simply truncating a draw from $\prob_v$. This family is compatible with the setting in Problem~\ref{prob:dpco}.

\begin{corollary}[Minimax risk lower bound, $\bP_{G, p}$]\label{cor:mmrisk_truncate}
In the setting of Proposition~\ref{prop:mmrisk}, 
\[\eps_k^\star \Par{\bbA_k \cup \bbB_k, \bP_{G, p}, \xset} = \Omega\Par{GD\max\Par{1 - \frac 1 p, \frac 1 {\log d}}\min\Par{1, \sqrt{\frac d {k\log (dk)}}}}.\]
\end{corollary}
\begin{proof}
We define a distribution $\overline{\prob}_v$ as follows: first $s \sim \prob_v$, and then if $\norm{s}_q \ge G$, we set $s \gets 0$. By adjusting the logarithmic term in \eqref{eq:sigmachoice} to be $O(\log(dk))$, with probability at most $\text{poly}((dk)^{-1})$, all $k$ draws from $\prob_v$ and $\overline{\prob}_v$ used are identical by a union bound. Further, due to problem constraints the function error is always at most $GD$. So, the risk is affected by at most $GD \cdot \text{poly}((dk)^{-1})$.
\end{proof}

Corollary~\ref{cor:mmrisk_truncate} shows that when $\beta$ in Assumption~\ref{assume:warmstart} is polynomially bounded, the value oracle complexities used by Theorem~\ref{thm:lp} for both DP-SCO and DP-ERM are near-optimal up to a quadratic factor overhead for the expected excess risk bounds they produce, even without requiring privacy. We believe this result shows the promise of our new framework, and provides strong motivation for closing the gaps in our mixing time bounds stated in Section~\ref{sec:future}.

Finally, we show that the value oracle complexity of our sampler in Theorem~\ref{thm:mainllt} is also near-optimal up to a quadratic factor. Again, removing this quadratic overhead would yield an optimal framework for non-Euclidean proximal sampling up to logarithmic factors, in light of Corollary~\ref{cor:sample_lowerbound}.

\begin{corollary}\label{cor:sample_lowerbound}
In the setting of Proposition~\ref{prop:mmrisk}, let $r: \xset \to \R$ be $1$-strongly convex in $\norm{\cdot}_p$ with additive range $O(D^2 \min(\log d, \frac 1 {p - 1}))$. Let $\ind$ be a distribution over $i$ such that all $f_i: \xset \to \R$ are $G$-Lipschitz in $\norm{\cdot}_p$, and let $F \defeq \E_{i \sim \ind} f_i$. No algorithm using $o(\frac{G^2}{\mu} \log^{-4} d)$ value oracle queries to some $f_i$ samples within total variation \[o\Par{\min\Par{\frac 1{\log d}, \sqrt{\frac{d}{k\log^3(dk)}}}}\]
of the density proportional to $\exp(-F - \mu r(x)) \1_{\xset}(x)$.
\end{corollary}
\begin{proof}
Assume for contradiction that $\alg$ is an algorithm satisfying the stated criterion using $k = o(\frac{G^2}{\mu} \log^{-4} d)$ value oracle queries, and let $F$ be minimized by $x^\star \in \xset$. We choose 
\[\mu = \frac{d}{D^2\min(\log d, \frac 1 {p - 1})}.\]
 Lemma~\ref{lem:sampling_error_d} then shows that the sampled $x$ satisfies
\begin{align*}\E_{x \sim \alg}\Brack{F(x)} - F(x^\star) &\le \mu\Par{r(x^\star) - r(x)} + d + GD \cdot o\Par{\min\Par{\frac 1{\log d}, \sqrt{\frac{d}{k\log^3(dk)}}}} \\
	&= O(d) + o\Par{\frac{GD}{\log d}\min\Par{1, \sqrt{\frac{d}{k\log(dk)}}}}.\end{align*}
For the given values of $k$ and $\mu$, this contradicts Corollary~\ref{cor:mmrisk_truncate}.
\end{proof}

Corollary~\ref{cor:sample_lowerbound} implies that for samplers with value query complexity depending polylogarithmically on the total variation distance, $\frac{G^2}{\mu}$ queries are required (up to polylogarithmic factors). This applies to the setting of our sampler in Theorem~\ref{thm:mainllt}; we also note that the LLT-based regularizers we use in our $\ell_p$ applications (Section~\ref{ssec:privacy}) satisfy the additive range bound in Corollary~\ref{cor:sample_lowerbound}.

%
\section{Lower bound on the range of $\psi_{1, 1}$}\label{app:sizebound}
In this section, we provide a lower bound on the range of $\psi_{1, 1}$ \eqref{eq:llt_lp} which grows with the dimension $d$, demonstrating non-scale invariance of our family of LLTs. Recall that $\psi_{1, 1}(x)$ is defined by 
\[\psi_{1, 1}(x) \defeq \log\Par{\int \exp\Par{\inprod{x}{y} - \norm{y}_\infty^2}\dd y}. \]

\begin{lemma}
	The additive range of $\psi_{1, 1}$ over $\{x\in \R^d\mid \norm{x}_1\leq 1\}$ is $\Omega(\sqrt d)$.
\end{lemma}
\begin{proof}
Throughout the proof denote for simplicity $\psi := \psi_{1,1}$ and let 
\[\Npsi_x(y) \propto \exp\Par{\inprod{x}{y} - \norm{y}_\infty^2}.\]
Then, following \eqref{eq:psirange}, we can write $\psi(x)- \psi(0)$ as 
\[\psi(x) - \psi(0) = \log \Brack{\E_{y\sim \Npsi_0}\exp (\inprod x y)},\]
where $\Npsi_0 \propto \exp(-\norm{y}_\infty^2)$.
Let $\pi$ be the probability density on $\R_{\ge 0}$ such that 
\[\dd \pi(r) \propto r^{d-1}\exp(-r^2)\dd r.\]
Here, $\dd \pi(r)$ is the density of the scalar quantity $r = \norm{y}_\infty$ for $y\sim \Npsi_0$. Note that the distribution of $y$ conditioned on $\norm{y}_\infty = r$ is uniform over the surface of the $\ell_\infty$ ball, where one random coordinate is set to $\pm r$, and the remaining coordinates are uniform on a $d-1$ dimensional hypercube with side length $r$. We denote this distribution as $\prob_r$, and write
\begin{align*}
\E_{y\sim \Npsi_0}\exp (\inprod x y) &= \E_{r\sim \pi} \Brack{\E_{y\sim \prob_r} \exp (\inprod x y)} \\
& = \E_{r\sim \pi} \Brack{\frac 1 d \sum_{i^\star \in [d]} \frac 1 2\sum_{y_{i^\star}\in \{-r, r\}} \exp(x_{i^\star} y_{i^\star}) \prod_{i\neq i^*}\int_{-r}^r  \frac 1 {2r}\exp ( {x_i} {y_i})\dd y_i} .
\end{align*}
Let $x =e_1$ and $g_{i^\star}^{(r)}=  \exp(x_{i^\star} r) \prod_{i\neq i^\star}\int_{-r}^r  \frac 1 {2r} \exp ( {x_i} {y_i})\dd y_i$. Then, 
\[\E_{y\sim \Npsi_0}\exp (\inprod x y)\geq \frac 1 {2d} \sum_{i^\star \in [d]} \E_{r\sim \pi(r)} g_{i^\star}^{(r)}\]
since this drops terms where $y_{i^\star} = -r$. When $i^\star= 1$, we have $g_{i^\star}^{(r)} = \exp(r)$.  When $i^\star \neq 1$, we have 
\begin{align*}
	g_{i^\star}^{(r)} = \int_{-r}^r \frac 1 {2r}\exp (y_{1})\dd y_1 = \frac {1}{2r}\Par{\exp(r) - \exp(-r)} .
\end{align*}

Now, consider $r_1 = \sqrt {\frac {d-1}{2}}$. For any $r \leq r_1$, $\frac {\dd}{\dd r}[ (d-1)\log r - r^2] = \frac {d-1}{r}-2r \geq 0$. Thus, we have
\begin{equation}\label{eq:rangelb_eq}
I:= \int_{0}^{\frac 1 2 r_1}{\exp((d-1)\log r  -r^2)\dd r} \leq \int_{\frac 1 2 r_1}^{r_1}{\exp((d-1)\log r  -r^2)\dd r}.\end{equation}
Letting $Z \defeq \int_{0}^{\infty}{\exp((d-1)\log r  -r^2)\dd r} $, \eqref{eq:rangelb_eq} shows that $$\int_{\frac 1 2 r_1}^{\infty}{\exp((d-1)\log r  -r^2)\dd r}= Z-I \geq Z-\frac 1 2 Z = \frac 1 2 Z.$$
Then,  for all $i^\star \in [d]$,
\begin{align*}
	\E_{r\sim \pi} g_{i^\star} &  =  \frac{\int_{0}^{\infty}{\exp((d-1)\log r -r^2 ) g_{i^\star}^{(r)}} \dd r} {Z}\\	&  \geq \frac{\int_{\frac 1 2 r_1}^{\infty}{\exp((d-1)\log r -r^2)g_{i^\star}^{(r)}}  \dd r}{Z}\\
	& \geq \frac{2\int_{\frac 1 2 r_1}^{\infty}{\exp((d-1)\log r -r^2)g_{i^\star}^{(r)}}  \dd r}{\int_{\frac 1 2 r_1}^{\infty}{\exp((d-1)\log r  -r^2)\dd r} } \\
	& \geq 2\min_{r\geq r_1}\exp(r-\log (4r)) = 2\exp(r_1 - \log(4r_1)).
\end{align*}
The fourth step follows from $g_{i^\star}^{(r)} \geq \frac 1 {4r} \exp( r)$ for $r \geq r_1$. The last step follows from $r - \log 4r$ increases on $r\geq r_1$. 
Combining with $\E_{y\sim \prob_0}\exp (\inprod x y)\geq \frac 1 {2d} \sum_{i^\star \in [d]} \E_{r\sim \pi(r)} g_{i^\star}$,
\begin{align*}
	\psi(x) - \psi(0)  = \log\E_{y\sim \prob_0}\exp (\inprod x y)  \geq \log \Par{\frac{d - 1}{d} \exp(r_1 - \log(4r_1))} = \Omega(\sqrt d).
\end{align*}
\end{proof}

\section{Deferred proofs from Section~\ref{sec:sampling}}\label{app:rejecthelper}

\restatelargebound*

\begin{proof}
	Clearly, it suffices to show $\E |\lam| \le \frac \delta 4 $. 
	Define random variables,
	\[\Delta_i \defeq |f_i(x_2) - f_i(x_1)|,\; \Delta \defeq \E_{i \sim \ind} \Delta_i,\]\
	whose randomness comes from $x_1,x_2 \sim \gamma_y$.
	By definition,
	\[\E |\lam| = \sum_{b > H} \frac 1 {b!} \E_{x_1, x_2 \sim \gamma}[\Delta]^B.\]
	
	Define $\Phi(t):=\sum_{b>H}\frac{t^b}{b!}$.
	For $H=\lceil 10\log \frac 1 \delta\rceil$, it is straightforward to check $\Phi(t)\le\frac{\delta}{16}$ for any $|t|\le 1$, and for all nonnegative $t$, $\Phi(t) \le \exp(t)$.
	Hence, letting $p_\Delta$ be the density of $\Delta$,
	\begin{equation}\label{eq:splitbyceiling}
		\begin{aligned}
			\E|\lambda| &\le \frac{\delta}{16}+\E[ \1_{\Delta>1}e^{\Delta}] \le \frac{\delta}{16} + \int_1^\infty \exp\Par{\lceil \Delta \rceil} p_\Delta(\Delta) \dd \Delta \\
			&\le \frac \delta {16} + \sum_{k\ge 1}\exp(k+1)\Pr_{x_1,x_2\sim\gamma}[\Delta\ge k].
		\end{aligned}
	\end{equation}
	It now suffices to bound on $\Pr[\Delta\ge k]$.
	Define a function $h_{x_1,x_2}(k):=\Pr_{i\sim \ind}[|f_i(x_1)-f_i(x_2)|\ge k]$.
	Since each $f_i$ is $G$-Lipschitz, and $\gamma_y$ is $\frac{1}{12\eta}$-strongly logconcave in by Lemma~\ref{lem:scllt}, by Lemma~\ref{lem:sc_conc}:
	\begin{align*}
		\E_{x_1,x_2}[h_{x_1,x_2}(k)]=\Pr_{x_1,x_2,i \sim \ind}[|f_i(x_1)-f_i(x_2)|\ge k]\le 4\exp\Par{-\frac{k^2}{96\eta G^2}},
	\end{align*}
	and so by Markov's inequality we have
	\begin{equation}\label{eq:boundh}
		\Pr_{x_1,x_2}[h_{x_1,x_2}(k)\ge e^{-t}]\le 4\exp\Par{t-\frac{k^2}{96\eta G^2}}.
	\end{equation}
	For fixed $x_1,x_2$, as each $f_i$ is $G$-Lipschitz in $\normx{\cdot}$, $|f_i(x_1)-f_i(x_2)|\le G\normx{x_1-x_2}$, and hence
	\begin{align*}
		\E_{i\sim \ind}[|f_i(x_1)-f_i(x_2)|]&\le \min_{k \ge 0}  k+ h_{x_1,x_2}(k) \cdot G\normx{x_1-x_2}.
	\end{align*}
	This then shows that if for some $k$, $h_{x_1,x_2}(k)\le \exp(-\frac{k^2}{192\eta G^2})$,
	\[
	\E_{i\sim \ind}[|f_i(x_1)-f_i(x_2)|]\le k+ \exp\Par{-\frac{k^2}{192\eta G^2}}\cdot G\normx{x_1-x_2},
	\]
	which implies via \eqref{eq:boundh} that
	\begin{equation}\label{eq:smallnorm}
		\begin{aligned}
			&~\Pr_{x_1,x_2}\Brack{\Delta \ge k+ \exp\Par{-\frac{k^2}{192\eta G^2}} \cdot G\normx{x_1-x_2}}\\
			\le&~ \Pr_{x_1,x_2}\Brack{h_{x_1,x_2}(k)\ge \exp\Par{-\frac{k^2}{192\eta G^2}}} \le 4\exp\Par{-\frac{k^2}{192\eta G^2}}.
		\end{aligned}
	\end{equation}
	Further, since $\normx{x_1-\E x_1}$ is a $1$-Lipschitz function in $x_1$ with a nonnegative mean, by Lemma~\ref{lem:sc_conc},
	\begin{equation}\label{eq:bignorm}
		\begin{aligned}
			\Pr\Brack{\normx{x_1-x_2}\ge k} \le 2\Pr\Brack{\normx{x_1-\E x_1}\ge k} \le 2\exp\Par{-\frac{k^2}{96\eta G^2}}.
		\end{aligned}
	\end{equation}
	Combining \eqref{eq:smallnorm} and \eqref{eq:bignorm}, 
	\begin{equation}\label{eq:boundbyceiling}
		\begin{aligned}
			\Pr_{x_1,x_2}[\Delta\ge 2k]
			&= \Pr_{x_1,x_2}\Brack{\Delta \ge 2k\wedge \normx{x_1-x_2}\ge \frac k G}+\Pr_{x_1,x_2}\Brack{\Delta \ge 2k\wedge \normx{x_1-x_2}\le \frac k G} \\
			&\le  2\exp\Par{-\frac{k^2}{96\eta G^2}} +\Pr_{x_1,x_2}\Brack{\Delta \ge k+\exp\Par{-\frac{k^2}{192\eta G^2}}G\normx{x_1-x_2}} \\
			&\le 6\exp\Par{-\frac{k^2}{192\eta G^2}}.
		\end{aligned}
	\end{equation}
	Plugging \eqref{eq:boundbyceiling} into \eqref{eq:splitbyceiling}, and using $\eta^{-1} \ge 10^4 G^2 \log \frac 1 \delta$, we have the desired
	\begin{align*}
		\E(|\lambda|\1_{\rho\notin[0,2]})\le \frac{\delta}{16}+\sum_{k= 1}^{\infty}6\exp\Par{k-\frac{k^2}{768\eta G^2}}\le \frac \delta 4.
	\end{align*}

\end{proof}

\restatesmallbound*

\begin{proof}
	We begin by bounding, analogously to \eqref{eq:splitbyceiling},
	\begin{align}
		\label{eq:decomp_sigma}
		\E[|\sigma|\1_{\rho\notin[0,2]}]\le 2^H\Pr[\rho\notin[0,2]]+\sum_{k \ge 1} \Pr\Brack{|\sigma| > 2^{kH}}2^{(k+1)H}.
	\end{align}
	
	Recall when $a \le H$, $|\jset| \le \half H^2$. By a union bound over Lemma~\ref{lem:sc_conc}, 
	\begin{align*}
		\Pr_{x_1,x_2}\Brack{|f_i(x_1)-f_i(x_2)|\ge \frac{2^k} 3\;\forall i \in \jset}\le H^2\exp\Par{-\frac{4^k}{864\eta G^2}}.
	\end{align*}
	If for each $i\in \jset$, $|f_i(x_1)-f_i(x_2)|\le \frac{2^k} 3$, we have for $k \ge 1$
	\begin{align*}
		|\sigma|=\sum_{b = 0}^H \1_{a \ge b} \prod_{i \in [b]} (f_{\jia}(x_2) - f_{\jia}(x_1))\le 1+\sum_{b=1}^{H}\Par{\frac{2^k}{3}}^{b}\le 2^{kH},
	\end{align*}
	which implies that $\Pr[|\sigma|\ge 2^{kH}]\le H^2\exp(-\frac{4^k}{864\eta G^2})$ and hence using our choice of $\eta\le \frac{1}{500G^2H}$,
	\begin{equation}\label{eq:large_sigma}
		\begin{aligned}
			\sum_{k=1}^{\infty}2^{(k+1)H}\Pr\Brack{|\sigma|>2^{kH}}&\le \sum_{k=1}^{\infty}2^{(k+1)H}H^2\exp\Par{-\frac{4^k}{864\eta G^2}} \\ 
			&\le\sum_{k=1}^{\infty}2^{4kH}\exp(-2\cdot 4^kH)\le \sum_{k=1}^{\infty}2^{-kH}\le \frac \delta 8.
		\end{aligned}
	\end{equation}
	It remains to bound $\Pr[\rho\notin[0,2]]$.
	Recall $\Pr[a > H] \le \frac 1 {H!}$ so since $a \le H \implies \sigma = \rho$, $\Pr[\rho\notin[0,2]]\le\frac{1}{H!}+\Pr[\sigma\notin[0,2]]$.
	Next, by a union bound over Lemma~\ref{lem:sc_conc} and $\half H^2$ indices in $\jset$,
	\begin{align*}
		\Pr_{x_1,x_2}\Brack{|f_i(x_1)-f_i(x_2)|\ge \half\; \forall i \in \ind}\le 2H^2\exp\Par{-\frac{1}{384\eta G^2}}.
	\end{align*}
	Under the event that $|f_i(x_1)-f_i(x_2)|\le \half$ for all $i \in \ind$, $0\le \sigma\le 2$ by definition. Hence we know $\Pr[\sigma\notin[0,2]]\le 2H^2\exp(-\frac{1}{384\eta G^2})$ and by our setting that $H>10\log \frac 1 \delta$, we have
	\begin{align}
		\label{eq:large_rho}
		\Pr[\rho\notin[0,2]]\cdot 2^H\le 2^H\Par{2H^2\exp\Par{-\frac{1}{384\eta G^2}}+\frac{1}{H!}}\le \frac \delta 8.
	\end{align}
	Combining~\eqref{eq:decomp_sigma},~\eqref{eq:large_sigma} and~\eqref{eq:large_rho} completes the proof.
\end{proof} 	\end{appendix}

\end{document}